\newif\ifSpringer

\newif\ifFinal

\newif\ifBlind

\Springertrue
\Springerfalse

\Finaltrue

\Blindfalse

\ifSpringer
\documentclass[runningheads]{llncs}
\else
\documentclass{article}
\fi

\usepackage{graphicx} 
\usepackage[font=footnotesize,labelfont=bf]{caption}
\usepackage[utf8]{inputenc}
\usepackage[T1]{fontenc}
\usepackage{amsmath,amssymb,amsfonts,bm}
\usepackage[hypertexnames=false]{hyperref}
\usepackage[inline]{enumitem}
\usepackage{listings}
\usepackage{xcolor}
\usepackage{graphicx}
\usepackage{wrapfig}
\usepackage[subrefformat=parens,labelformat=parens]{subcaption}
\usepackage{cite}
\usepackage[ruled,linesnumbered,lined,noend]{algorithm2e}
\usepackage{multirow}
\usepackage{array}
\usepackage{multicol}
\usepackage{longtable}
\usepackage{booktabs}
\usepackage{orcidlink}
\usepackage{xparse}

\ifSpringer
\else
\usepackage{amsthm}
\fi

\AtBeginDocument{\setlength\abovedisplayskip{4pt plus 2pt minus 2pt}}
\AtBeginDocument{\setlength\belowdisplayskip{\abovedisplayskip}}

\setlist[itemize]{noitemsep, topsep=4pt plus 2pt minus 2pt}

\DeclareMathOperator*{\argmax}{argmax}

\hypersetup{
    colorlinks=true,
    linkcolor=blue,
    filecolor=magenta,
    urlcolor=cyan,
    citecolor=purple,
}

\ifSpringer
\else
\newcommand*{\inst}[1]{\textsuperscript{#1}\:}

\theoremstyle{plain}
\newtheorem{theorem}{Theorem}
\theoremstyle{plain}
\newtheorem{definition}{Definition}
\theoremstyle{definition}
\newtheorem{example}{Example}
\fi

\newcommand{\Hide}[1]{}

\newcommand{\ColorEmph}[1]{{\color{blue} #1}}

\newcommand{\Todo}[1]{\ifFinal\else{\em \color{red} ++#1}\fi}
\newcommand{\NEW}[1]{\ifFinal#1\else{\color{brown} #1}\fi}
\newcommand{\TodoTK}[1]{\Todo{\textbf{Tomas}: #1}}
\newcommand{\TodoFL}[1]{\Todo{\textbf{FL}: #1}}
\newcommand{\gri}[1]{\Todo{\textbf{Grigory}: #1}}
\newcommand{\ParagraphTitle}[1]{#1.}
\newcommand{\Paragraph}[1]{%
\ifSpringer\subsubsection*{\ParagraphTitle{#1}}
\else\paragraph*{\ParagraphTitle{#1}}
\fi
}
\newcommand{\SubParagraph}[1]{%
\ifSpringer\paragraph*{\ParagraphTitle{#1}}
\else\subparagraph*{\ParagraphTitle{#1}}
\fi
}

\newcommand*{\Concat}{\mathbin\Vert}

\newcommand*{\Unsat}{\texttt{unsat}}
\newcommand*{\Sat}{\texttt{sat}}

\newcommand*{\AlgVar}[1]{\ensuremath{\mathit{#1}}}

\SetArgSty{upshape}
\SetFuncSty{textsc}
\SetFuncArgSty{upshape}

\SetKwInput{Global}{Global}

\SetKw{Const}{const }
\SetKw{Break}{break}
\SetKw{Continue}{continue}
\SetKw{Let}{let}
\SetKw{In}{in}
\SetKw{Assert}{assert}

\SetKwProg{Function}{function}{}{}
\SetKwBlock{Loop}{loop}{end}

\SetKwComment{CommentBase}{$\triangleright$\ }{}
\SetCommentSty{text}

\NewDocumentCommand{\Comment}{s o m}{%
  \IfBooleanTF{#1}{%
    \hspace{0pt plus 1filll}%
    \IfValueTF{#2}{\CommentBase*[#2]{#3}\hspace{-\skiptext}}{\CommentBase*[h]{#3}\hspace{-\skiptext}\;}%
  }{%
    \IfValueTF{#2}{\CommentBase[#2]{#3}}{\CommentBase{#3}}%
  }%
}

\newcommand*{\QED}{\ifSpringer\hfill $\square$\fi}

\newcommand*{\Relu}{\ensuremath{\mathit{ReLU}}}

\newcommand*{\Expl}[1][]{\ensuremath{\varphi_{#1}}}

\newcommand*{\Classifier}{\ensuremath{\mathcal{M}}}
\newcommand*{\ClassifierTuple}{\ensuremath{(\FSet, \FSpace, \ClSet, \ClFun)}}

\newcommand*{\FSet}{\ensuremath{\mathcal{F}}}
\newcommand*{\ClSet}{\ensuremath{\mathcal{K}}}
\newcommand*{\ClFun}{\ensuremath{\kappa}}

\newcommand*{\VertexSet}{\ensuremath{\mathcal{V}}}
\newcommand*{\EdgeSet}{\ensuremath{\mathcal{E}}}

\newcommand*{\Hidden}{\ensuremath{\mathcal{H}}}
\newcommand*{\HiddenNSet}{\ensuremath{\VertexSet_{\Hidden}}}

\newcommand*{\Domain}[1][]{\ensuremath{\mathcal{D}_{#1}}}

\newcommand*{\FSpace}{\ensuremath{\mathbb{F}}}

\newcommand*{\ImpSpace}[1][\Expl]{\ensuremath{\FSpace{}_{#1}}}

\newcommand*{\FVec}{\ensuremath{\mathbf{x}}}

\NewDocumentCommand{\act}{O{\ell} O{i}}{%
  \ensuremath{a_{#2}^{(#1)}}%
}
\NewDocumentCommand{\zed}{O{\ell} O{i}}{%
  \ensuremath{z_{#2}^{(#1)}}%
}
\NewDocumentCommand{\neuron}{O{\ell} O{i}}{%
  \ensuremath{v_{#2}^{(#1)}}%
}

\newcommand*{\ClassifierPhi}{\ensuremath{\psi_{\Classifier}}}
\newcommand*{\DomainPhi}{\ensuremath{\psi_{\Domain}}}
\newcommand*{\ClPhi}[1][c]{\ensuremath{\psi_{#1}}}

\newcommand*{\Sample}{\ensuremath{\mathbf{s}}}
\newcommand*{\SampleExpl}{\Expl[\Sample]}

\newcommand*{\Itp}{\ensuremath{\mathit{Itp}}}

\newcommand*{\ItpType}[1]{\ensuremath{\Itp{}_\mathit{#1}}}

\newcommand*{\ItpF}{\ItpType{F}}
\newcommand*{\ItpDF}{\ItpType{DF}}
\newcommand*{\ItpFactor}{\ItpType{f}}

\newcommand*{\aItpAlg}[1]{\texttt{#1}}
\newcommand*{\astronger}{\aItpAlg{stronger}}
\newcommand*{\astrong}{\aItpAlg{strong}}
\newcommand*{\amid}{\aItpAlg{mid}}
\newcommand*{\aweak}{\aItpAlg{weak}}
\newcommand*{\aweaker}{\aItpAlg{weaker}}

\newcommand*{\Logic}[1]{\textsf{#1}}
\newcommand*{\QFLRA}{\Logic{QF\_LRA}}
\newcommand*{\QFFP}{\Logic{QF\_FP}}

\newcommand*{\Formula}{\ensuremath{\phi}}
\newcommand*{\GuidePhi}[1][]{\ensuremath{\gamma_{#1}}}
\newcommand*{\GuideVar}[1][]{\ensuremath{g_{#1}}}
\newcommand*{\GuideN}{\ensuremath{n_{\GuideVar}}}

\newcommand*{\Tool}[1]{\textsc{#1}}
\newcommand{\Opensmt}{\Tool{OpenSMT}}
\newcommand{\Zthree}{\Tool{Z3}}
\newcommand{\Mathsat}{\Tool{MathSAT5}}

\newcommand{\Verix}{\Tool{VeriX}}
\newcommand{\Spexplain}{\Tool{S\textsubscript{p}EX\textsubscript{pl}AI\textsubscript{n}}}

\newcommand*{\IntConf}{\ensuremath{\mathcal{R}_{\Hidden}}}
\newcommand*{\GuideConf}{\ensuremath{\mathcal{G}_{\Hidden}}}
\newcommand*{\FixConf}{\ensuremath{\mathcal{F}_{\Hidden}}}
\newcommand*{\ConfigPhi}{\ensuremath{\Expl[\mathcal{C}]}}

\newcommand{\partwo}{PAR-2}

\newcommand{\naive}{\emph{Direct}}
\newcommand{\fixall}{\emph{Fix-all}}
\newcommand{\guideall}{\emph{Guide-all}}
\newcommand{\hybrid}{\emph{Hybrid}}

\date{}
\title{
    Neuron Activation-based
    Computation of
    Logical Explanations
    for Deep Neural Networks%
}

\author{\ifBlind\else%
Tomáš Kolárik\inst{1}\orcidlink{0000-0002-7207-5197},
Faezeh Labbaf\inst{1}\orcidlink{0000-0002-8812-6702},
Fabrizio Leopardi\inst{1}\orcidlink{0009-0004-8370-3058},
\ifSpringer\else\\\fi
Grigory Fedyukovich\inst{1,3}\orcidlink{0000-0003-1727-4043},
Michael Wand\inst{2}\orcidlink{0000-0003-0966-7824},
Natasha Sharygina\inst{1}\orcidlink{0000-0002-8872-4913}%
\ifSpringer
\else
\\[1em]
\small
\inst{1} University of Lugano (USI), Lugano, Switzerland
\\
\small
\inst{2} SUPSI, IDSIA, Lugano, Switzerland
\\
\small
\inst{3} Florida State University, United States of America
\fi
\fi}

\ifSpringer
\authorrunning{\ifBlind\else%
Kolárik T., Labbaf F., Leopardi F., Fedyukovich G., Wand M., Sharygina N.
\fi}

\institute{\ifBlind\else%
University of Lugano (USI), Lugano, Switzerland
\and
SUPSI, IDSIA, Lugano, Switzerland
\and
Florida State University, United States of America%
\fi}
\fi 

\newcommand{\Acknowledgement}{%
\SubParagraph{Acknowledgements}
\ifSpringer
\else
This is an extended version of the paper published at ATVA'26.
\fi
This work was conducted as part of the
``Formal Reasoning on Neural Networks'' project
funded by the Hasler Foundation, Switzerland.
}

\begin{document}

\maketitle

\begin{abstract}
    Formal explainability of classifying neural networks (NNs) is an active area of research,
    providing explanations with provable guarantees of the classification
    within continuous regions of the input feature space.
    However,
    the existing techniques are either limited to individual input features
    without guarantees on their relations
    or
    the provided solutions fail to scale to deep architectures.
%
    This paper addresses these issues
    by introducing a flexible symbolic framework
    for an efficient, guided computation of explanations of the NN behavior,
    parametrized by the activations of internal neurons,
    and
    using logical engines such as SMT solvers.
%
    Unlike prior methods that rely on specialized NN verifiers, our method yields explanations that are not restricted in shape.
    Our algorithm is implementable on top of a general-purpose logical solver,
    isolating the NN-specific encoding from the algorithmic framework.
%
    We experimented with a wide range of benchmarks
    from the domains of image recognition and medicine,
    illustrating the advantages of the new method,
    particularly in computational efficiency.
    Notably,
    our approach enables
    logical explanation
    of deep networks
    not amenable to prior logic-based methods.

\end{abstract}

\section{Introduction}

A recent trend in explaining
the behavior of classifying neural networks (NNs)
is based on formal methods and first-order logic:
Given a NN
and
a sample point,
that is,
an assignment of all input features of the NN
to concrete values,
the problem is to
compute a logical justification
in the form of a sufficient condition
for the NN classification.
Logic-based methods
encode the NN behavior and the classification outcome
into a logical formula
and then reason over it
while generalizing the constraints on the sample point.
Although these methods
deliver formal explanations accompanied by \emph{provable guarantees}
regarding their correctness
and
can
precisely capture the semantic behavior of the NN
within the vicinity of the sample point,
their \emph{computation} is a significant \emph{bottleneck}%
~\cite{LabbafKBFWS:25,Wu_Li_Wu_Barrett_2026,wu2023verixverifiedexplainabilitydeep,ignatiev2019abduction,Shih_IJCAI18_SymbolicApproach,Seshia_TowardsVerifiedAI,Marques-Silva2023,LaMalfa_IJCAI21_GuaranteedOptimalRobustExplanations,Ignatiev:24,bassan2025explaining},
not amenable to realistic deep architectures.
Early logic-based methods%
~\cite{Wu_Li_Wu_Barrett_2026,wu2023verixverifiedexplainabilitydeep,ignatiev2019abduction,Shih_IJCAI18_SymbolicApproach,Seshia_TowardsVerifiedAI,Marques-Silva2023,LaMalfa_IJCAI21_GuaranteedOptimalRobustExplanations,Ignatiev:24,bassan2025explaining}
rely on symbolic engines
that are based on constraint or SMT (Satisfiability Modulo Theories~\cite{DeMoura_SMTIntroduction,Nieuwenhuis:2006}) solving,
like
SMT solvers~\cite{cvc5,z3,opensmt,mathsat5,yices2,smtinterpol}
or
specialized
NN verifiers
such as
Marabou~\cite{Katz_CAV19_MarabouFramework}
to derive explanations via \emph{abductive reasoning},
which originates from program verification~\cite{DilligDLM13,Albarghouthi2016,Prabhu:Abduction1,Prabhu:Abduction2}.
These methods
iteratively relax input constraints by enumerating the input features.
This process
is highly sensitive to feature count,
requiring a series of computationally expensive queries to a black-box solver,
and
depends heavily on selecting
a specific feature order.
\NEW{%
The approach in~\cite{bassan2025explaining} targets the query overhead
by accelerating individual verification calls
by abstracting and iteratively refining the original NN,
but the sensitivity to feature count remains a bottleneck.
}%
In addition to the computational complexity of solving,
often similar to an unsatisfiable core extraction,
the issue is that
the resulting explanations exhibit hyperrectangular shapes
that stem from the formula structure
and
fail to capture input feature relations
or to scale beyond specific sample points.
Furthermore,
the encoding and reasoning over the NN are often coupled:
while the symbolic solvers are heavily optimized tools,
they could be constrained by such an approach.
%
To address these problems,
a new method~\cite{LabbafKBFWS:25,KolarikFLLSW:26} for computing generic logical explanations
that isolates the NN-specific encoding from the reasoning algorithms
was recently proposed.
The explanations are derived from proofs of unsatisfiability,
preserving complex relationships between input features,
and utilizing the full power of symbolic solvers
without any NN-specific alterations to them.
For example,
\cite{LabbafKBFWS:25}
utilizes Craig interpolation~\cite{Cra57}
and unsatisfiable core extraction
to
extract formal justifications
for classification,
removing the need for multiple calls to the underlying solver.
While these explanations are correct by construction, not restricted in shape, and offer greater generality than antecedent methods,
the approach still scales poorly.
Similarly to other prior methods,
the problem stems from the direct encoding of the NN
without sharing any information about its structure
or how it evaluates at the sample point.

To overcome the scalability limitations of all prior logical explanation approaches,
we introduce a mechanism that systematically focuses the computation
on the sample point
while preserving or even improving explanation generality in practice.
Our solution decouples and advances both the encoding and solving phases.
At the encoding level,
we introduce
    a comprehensive
    \emph{parametrizable NN encoding}
    that is driven by the activation of internal neurons
    at the given sample point,
    and by an input configuration specific to the NN
    that is either user-provided
    or automatically generated
    by external NN analysis algorithms.
    The configuration selects
    internal neurons that are firmly fixed to only one activation phase
    and neurons
    where the engine is guided to favor one activation phase over the other.
    This provides a high flexibility in the trade-off
    between the computational complexity
    and the potential size of the space associated with the resulting explanations,
    while
    not restricting their shape
    and
    remaining provably correct.
    We also provide an efficient choice of the NN configuration
    by
    \emph{identifying important neurons}
    that maintains a balance
    in the selection
    between fixing and guiding
    the internal nodes
    based on a gradient traversal
    from the sample point towards a different class,
    observing the difference between the neuron activations.

    At the solver level,
    we introduce
    a novel transferable mechanism to optimize the performance
    of general-purpose logical solvers such as SAT or SMT for a concrete application such as NN analysis,
    by involving
    \emph{guiding constraints}.
    Unlike traditional approaches that rely on rigid,
    general-purpose heuristics,
    our method uses guiding variables to encode preferred properties
    into the logical formula,
    allowing solvers to prioritize decisions systematically
    aligned with the model's domain-specific structure.
    The technique can be applied wherever some constraints are more likely to hold than others
    and in any logical solver handling non-determinism.
    By leveraging also the order of decision suggestions
    derived from the NN architecture,
    the search space is pruned more efficiently
    than by existing abductive reasoning tools.
    This contribution bridges the gap
    between theoretical correctness and practical applicability,
    offering a scalable solution for generating rigorous logical explanations
    in deep learning contexts.


Overall, our novel neuron activation-based framework for automated computation of logical explanations
is agnostic to the backend solver,
as well as
to the techniques
that generalize the constraints on the sample point,
such as Craig interpolation
(where various algorithms are available)~\cite{LabbafKBFWS:25},
unsatisfiable core extraction~\cite{Wu_Li_Wu_Barrett_2026,bassan2025explaining,wu2023verixverifiedexplainabilitydeep,LabbafKBFWS:25},
or
interval expansion~\cite{Ignatiev:24,LabbafKBFWS:25}.
\NEW{%
This enables fine-tuning the trade-off between the generality and interpretability
of logical explanations---%
properties discussed in~\cite{KolarikFLLSW:26}.
}%
In particular,
when instantiated with Craig interpolation,
the resulting explanations are not restricted in shape,
unlike those generated by prior methods relying on specialized NN verifiers.
\NEW{%
Furthermore,
the approach generalizes to any acyclic network architecture
(e.g., convolutional neural networks---CNNs, ResNets)
and activation function,
handling arbitrary piecewise-linear functions directly (by guiding or splitting disjunctions)
and linearly overapproximating non-piecewise-linear ones~\cite{zhang2018efficient}.
}%
Although beyond the scope of this paper,
the computational framework can easily
integrate important external components
(e.g., \cite{LabbafKBFWS:25,Leopardi26})
to provide various interpretations (e.g., visual) of the logical explanations---%
an active research field
driven by diverse applications
such as image recognition and medicine.
\NEW{%
Notably,
while the resulting explanations remain large and complex
(sharing the complexity of those in~\cite{LabbafKBFWS:25}),
such interpretability methods render useful visualizations
in a matter of seconds.
}

We implemented the guided approach
using the \Opensmt{}~\cite{opensmt2,opensmt} solver
and
experimented with a wide range of real-world benchmarks
across a spectrum of complexities
from the medical and image-recognition domains,
demonstrating the advantages of the new method---%
particularly regarding its computational efficiency.
Notably,
our approach enables the explainability of deep networks
not amenable to prior logic-based techniques.
Compared to the leading abduction-based explainability methods for NNs,
namely
the case study~\cite{LabbafKBFWS:25},
and
the tool \Verix{}\footnote{%
    The recently published tool \Verix{}+~\cite{Wu_Li_Wu_Barrett_2026}
    improves the sensitivity of \Verix{} to the number of input features
    but still lacks evaluation on deep architectures.
}~\cite{wu2023verixverifiedexplainabilitydeep},
which
delivers
sound and more compact explanations
than other non-logic-based state-of-the-art tools
such as Anchors~\cite{Ribeiro_AAAI18_Anchors} and LIME~\cite{Ribeiro_2016}
on the image classification datasets MNIST~\cite{mnist}
and GTSRB~\cite{Stallkamp-IJCNN-2011},
our experiments
show that
we not only provide yet more general explanations,
but also substantially outperform
both~\cite{LabbafKBFWS:25,wu2023verixverifiedexplainabilitydeep}
in these datasets,
especially on deep architectures.
The comparative experimentation
illustrates the improvement of the generality of the explanations
by applying interpretation techniques
from prior works~\cite{LabbafKBFWS:25,Leopardi26}.

As a secondary contribution,
we experiment with applying the parametrized NN encoding to an alternative modular framework centered on inductive invariant synthesis.
By modeling the network as a discrete-time transition system rather than a monolithic function,
we facilitate the discovery of predicates describing relationships among neurons in each layer.
Our proof-of-concept experiments demonstrate that leveraging internal neuron configurations significantly enhances formal NN analysis,
even when decoupled from traditional monolithic SMT encodings.
Overall, the performance improved drastically: when fixing neurons, benchmarks were solved even when they were timing out without the fixing.

\Paragraph{Related Work}

Logic-based methods formulate the explainability of NNs
using abductive reasoning~\cite{KolarikFLLSW:26}.
There is a multitude of such methods%
~\cite{ignatiev2019abduction,Shih_IJCAI18_SymbolicApproach,Seshia_TowardsVerifiedAI,Marques-Silva2023,LaMalfa_IJCAI21_GuaranteedOptimalRobustExplanations,bassan2025explaining},
culminating in highly performant
\Verix{} and \Verix{}+\cite{Wu_Li_Wu_Barrett_2026} frameworks.
As such, all these techniques
are based on specialized NN verifiers
and
focus only
on special cases of the logical concept of abduction,
restricting to a specific formula structure
or a specific neighborhood of a sample point.
As a consequence,
they are also sensitive to the number of input features.
Similarly,
the work in \cite{Ignatiev:24}
provides interval-based explanations,
yet still constrains individual features
that cannot reflect feature relationships
and hence approximate complex decision boundaries.
Overall,
these approaches suffer from not being general
due to the choice of tools, techniques, and the type of explanations,
and
the resulting limitations disallow the use of the newly introduced concept of fixing neuron activations
that requires the support of unrestricted explanations.
The recent case study~\cite{LabbafKBFWS:25}
experimented with using Craig interpolation techniques
and unsatisfiable cores
to relax input constraints
to guarantee the generalization
of the explanation.
Notably,
only one satisfiability query suffices for this purpose.
However,
the method fails to scale to multi-layer architectures
due to direct
encoding of the NN.
The approach in this paper,
while inspired by the generalization ideas of~\cite{LabbafKBFWS:25},
not only elevates
the computation of unrestricted logical explanations
to a generic symbolic overapproximating reasoning,
but more importantly,
addresses the scalability issues
by
parametrizing the explainability task
with the NN architecture.

Alongside explainability,
NN verification (VNN)
is an active field
focused on whether a network satisfies a given property
(e.g., local robustness).
Solving such general satisfiability queries
often requires exhaustive exploration of the feature space.
The related problems, solutions, and tools are thoroughly surveyed in~\cite{Kanav2025}.
Some VNN tools serve
as internal engines for early logic-based explainability methods,
but their iterative black-box use leads to significant complexity blow-ups.
Such tools could be integrated into our neuron activation-based framework
if extended to support guiding constraints,
but
they typically lack generalization techniques
(e.g., Craig interpolation),
which are crucial for explainability.

Classic (non-logic-based) approaches to NN explainability either perform
analysis at the \emph{unit (neuron) level}~\cite{Erhan_VisualizingDNNFeatures,Zeiler_VisualizingConvNets}, which however fails to capture global properties or interdependencies between units;
are \emph{gradient-based} methods~\cite{Simonyan_ICLR14_DeepInsideConvNets,Bach_LayerwiseRelevancePropagation}, which describe the network behavior around a particular sample;
or \emph{inversion} methods~\cite{Mahendran_CVPR2015_UnderstandingDeepImageRepresentations},
which provide global, but only approximate, views.
There is also a range of \emph{model-agnostic} methods~\cite{Baehrens_ExplainIndividualDecisions,Lundberg_NIPS17_InterpretModelPredictions},
including SHAP~\cite{10.5555/3722577.3722589}, LIME~\cite{Ribeiro_2016},
and Anchors~\cite{Ribeiro_AAAI18_Anchors},
but they often yield logically inconsistent explanations~\cite{Marques-Silva2023}.
All the classic methods above
rely on approximations with \emph{no strict guarantees}
or are \emph{sample-based}:
the behavior outside the sampled distribution remains unknown.

Post-processing of logical explanations
is an important step
for providing a human-understandable interpretation,
such as a visualization or formula simplification.
There is a whole body of work on interpretation techniques,
which is outside the scope of this paper~\cite{LabbafKBFWS:25,Leopardi26,Wu_Li_Wu_Barrett_2026,wu2023verixverifiedexplainabilitydeep,10.5555/3722577.3722589,Ribeiro_2016,Ribeiro_AAAI18_Anchors,Simonyan_ICLR14_DeepInsideConvNets,bassan2025explaining}.
The general nature of logical explanations
used in this paper
makes them
amenable to such techniques.
In our experiments,
we used
techniques
such as
the visualization of the feature space~\cite{LabbafKBFWS:25},
minimum norm point extraction,
and hypervolume estimation~\cite{Leopardi26}
for comparative analysis.

\ifSpringer
\else
\Acknowledgement{}
\fi

\section{Preliminaries}
\label{sec:background}

\Paragraph{Neural Network Classifiers}
\label{sec:background:classification}
A finite set of input \emph{features} is $\FSet = \{1, \dots, m\}$,
each associated with a discrete or continuous \emph{domain}
$\Domain[i] \subset \mathbb{R}$,
either a discrete set of real numbers
or a closed interval.
A \emph{classifier} is $\Classifier = \ClassifierTuple$,
where
$\FSpace = \Domain[1] \times \Domain[2] \times \dots \times \Domain[m]$
is the input \emph{feature space},
$\ClSet = \{c_1, c_2, \dots, c_K\}$
a finite set of $K$ classes,
$K \geq 2$,
and
$\ClFun : \FSpace \rightarrow \ClSet$
is a \textit{classification function}.
An arbitrary point in the feature space is
\FVec{} =  $(x_1, \dots, x_m) \in \FSpace$,
and each $x_i \in \Domain[i]$ is an input \emph{feature variable}.
A specific \emph{sample point} in the feature space
is \Sample{} = $(s_1,\dots,s_m) \in \FSpace$,
where each $s_i \in \Domain[i]$ is a constant.
For all classes,
at least one sample point must be classified into the class.
%
%
A (fully connected) \emph{classifying neural network} (NN) is a classifier
represented
by a graph $(\VertexSet,\EdgeSet)$ with
the set of vertices~\VertexSet{}
and
the set of edges~\EdgeSet{},
representing the neurons and their weighted connections, respectively.
\NEW{%
The approach in this paper also extends to other acyclic architectures (e.g., CNNs).
}%
Neurons~\neuron{} are partitioned into~$L+1$ layers $\VertexSet^{(0)},\dots,\VertexSet^{(L)}$
such that
each $\VertexSet^{(\ell)} = \{ \neuron \}_{i=1}^{n^{(\ell)}}$
contains~$n^{(\ell)}$ neurons.
Layer~$\VertexSet^{(0)}$ is the \emph{input} layer,
$\VertexSet^{(L)}$ is the \emph{output} layer,
layers $\VertexSet^{(1)},\dots,\VertexSet^{(L-1)}$ are \emph{hidden} layers,
and
$\HiddenNSet = \bigcup_{\ell=1}^{L-1} \VertexSet^{(\ell)}$
is the set of neurons in hidden layers.
For each $0 \leq \ell < L$,
every neuron $\neuron \in \VertexSet^{(\ell)}$ is connected to every neuron $\neuron[\ell+1][j] \in \VertexSet^{(\ell+1)}$ by an edge $(\neuron, \neuron[\ell+1][j]) \in \EdgeSet$,
which is assigned a weight $w_{i,j}^{(\ell)} \in \mathbb{R}$,
and for each $1 \leq \ell \leq L$,
every neuron \neuron{} is assigned a bias $b_i^{(\ell)} \in \mathbb{R}$.
Each input neuron~$\neuron[0]$
is associated with input feature variable~$x_i \in \Domain[i]$
and
each output neuron~$\neuron[L]$ to class $c_i \in \ClSet$.
%
%
%
Given a neuron \neuron,
$\zed : \FSpace \to \mathbb{R}$
is the \emph{pre-activation}
and
$\act : \FSpace \to \mathbb{R}$
is the \emph{activation},
defined as follows:
$\zed[0](\FVec{}) = \act[0](\FVec{}) = x_i$,
and
for hidden and output layers ($1 \leq \ell \leq L$):
\begin{align}
\label{eq:nn:zed}
    \zed(\FVec{}) &= \sum_{j=1}^{n^{(\ell-1)}} w_{j,i}^{(\ell-1)}\act[\ell-1][j](\FVec{}) + b_i^{(\ell)} \\
\label{eq:nn:act}
    \act(\FVec{}) &=  \sigma\bigl(\zed(\FVec{})\bigr)
\end{align}
where
$\sigma(z) = \Relu(z) = \max\{z,0\}$
in the case of hidden layers
and
$\sigma(z) = z$
in the case of output layers.
A neuron $\neuron \in \HiddenNSet$ is \emph{active} w.r.t. input $\FVec \in \FSpace$
if
$\zed(\FVec) > 0$,
and \emph{inactive} if $\zed(\FVec) \leq 0$.
\NEW{%
This can be extended to arbitrary piecewise-linear~$\sigma$
or an overapproximation~\cite{zhang2018efficient}.
}%
%
%
The classification is given by the output neuron with the maximum activation:
\(
\ClFun(\FVec) = c_i \Longleftrightarrow i = \argmax_{j}\{\act[L][j]\}.
\)

\Paragraph{SMT Encoding of Neural Networks}
\label{sec:encoding}
Using
an SMT-LIB~\cite{BarFT-RR-17} quantifier-free first-order logic,
given a class $c \in \ClSet$,
$\ClFun(\FVec) \neq c$ is encoded by
$\psi := \ClassifierPhi \land \DomainPhi \land \neg \ClPhi$:
\begin{itemize}
    \item[\ClassifierPhi:]
    The neural network~\Classifier{}:
    the constraints
    in~(\ref{eq:nn:zed}),
    with fixed weights and biases,
    and
    in~(\ref{eq:nn:act}),
    are encoded directly.
    The \Relu{} function can be encoded
    using the if-then-else (ITE) term of the SMT-LIB standard.
    \item[\DomainPhi:]
    All domains~\Domain[i] of the feature space,
    that is,
    $\FVec \in \FSpace$
    (i.e., $x_i \in \Domain[i]$ for all~$i \in \FSet$).
    \item[$\neg \ClPhi$:]
    The constraint on the outcome of the classification
    is \emph{not} class~$c$:
    the value of at least one output neuron~$\act[L](\FVec)$
    exceeds that of class $c$.
\end{itemize}
Functions \zed{} and \act{} are encoded via inlining
or auxiliary variables.
Since $K \geq 2$ and classes are not redundant,
$\psi$ is satisfiable.
Given a sample point $\Sample \in \FSpace$
s.t. $\ClFun(\Sample) = c$,
and
using
\(
    \SampleExpl := \bigwedge_{i \in \FSet}\ x_i = s_i
\),
formula $\SampleExpl \land \psi$
is \emph{unsatisfiable}.
We instantiate the SMT-LIB logic to \QFLRA{} (Quantifier-Free Linear Real Arithmetic).
\ifSpringer
\else
More details on the encoding follow in Appendix~\ref{apx:example}.
\fi
Alternative encodings like \QFFP{} use floating-point semantics,
but sound solving and interpolation are much more challenging
in such settings.

\Paragraph{Logical Explanations}
Our solution aims at general logic-based explanations.
Thus,
we define a general logical concept
using the notion of abductive reasoning\footnote{
    For vectors of variables $\mathbf{x}, \mathbf{y}$
    and
    formulas $T(\mathbf{x}, \mathbf{y}), O(\mathbf{x}, \mathbf{y})$,
    the \emph{abduction problem} is to find a formula $\varphi(\mathbf{x})$
    s.t.:
    \begin{enumerate*}[label=(\arabic*)]
        \item $\varphi(\mathbf{x}) \land T(\mathbf{x}, \mathbf{y}) \nRightarrow \bot$,
        and
        \item $\varphi(\mathbf{x}) \land T(\mathbf{x}, \mathbf{y}) \Rightarrow O(\mathbf{x}, \mathbf{y})$.
    \end{enumerate*}
}
that subsumes all prior abduction-based concepts%
~\cite{Ignatiev:24,ignatiev2019abduction,Shih_IJCAI18_SymbolicApproach,Seshia_TowardsVerifiedAI,Marques-Silva2023,Wu_Li_Wu_Barrett_2026,wu2023verixverifiedexplainabilitydeep,LaMalfa_IJCAI21_GuaranteedOptimalRobustExplanations,bassan2025explaining}\footnote{
    Our generalization is influenced by the recent case study~\cite{LabbafKBFWS:25} on the application of Craig interpolation for the explainability of classifying NNs.
    Theoretical comparison of~\cite{LabbafKBFWS:25} with other abduction-based concepts is available in~\cite{KolarikFLLSW:26}.
}.

\begin{definition}[Logical Explanation, Impact Space]
\label{def:spaceex}
    Given a classifier~$\Classifier = \ClassifierTuple$,
    and a class~$c \in \ClSet$,
    a \emph{logical explanation} of class~$c$
    is a satisfiable formula~\Expl{}
    such that
    $\forall\, \FVec \in \FSpace \mathbin{.} \Expl(\FVec) \implies \ClFun(\FVec) = c$.
    The \emph{impact space} of~\Expl{}
    is the set
    $\ImpSpace := \{ \FVec \in \FSpace \mid \Expl(\FVec) \} \subset \FSpace$.
\end{definition}
Logical explanations with the smallest impact space
correspond to a sample point $\Sample{} \in \FSpace$:
\SampleExpl{}
with
$\ImpSpace[\SampleExpl] = \{ \Sample \}$.
Logical explanations are sufficient conditions of classification to class~$c$:
they guarantee the classification for all points covered by the impact space.
Unlike many prior methods,
the concept
can be applied
both to global and local\footnote{
    Given a sample~\Sample{}
    (can also be randomly selected),
    \Expl{} s.t.
    $\Sample \in \ImpSpace$
    is a local explanation.
}
explanations
since it does not rely on a sample point.
The general logical representation of explanations
is inherently flexible,
making it applicable to a wide range of interpretation techniques---%
an active research area on its own
(e.g., \cite{Leopardi26,LabbafKBFWS:25,wu2023verixverifiedexplainabilitydeep})
driven by diverse, application-specific user needs.

Various symbolic techniques for the construction
of generalized logical explanations
can be employed,
such as
unsatisfiable core extraction,
interval expansion,
or
Craig interpolation\footnote{%
    Given an unsatisfiable formula $A \land B$, a \emph{Craig interpolant}\Hide{~\cite{Cra57}} is a formula $I$
    s.t.:
    $A \Rightarrow I$,
    $I \land B$ is unsatisfiable,
    and
    $I$ uses only the common variables of $A$ and $B$.
    $I = \Itp{}(A, B)$ denotes
    the interpolant~$I$ computed by an~interpolation algorithm $\Itp$
    from~$A$ and~$B$.
}.
A Craig interpolant $I = \Itp(\Expl, \psi)$,
where \Expl{} is a logical explanation and $\psi$ encodes $\ClFun(\FVec) \neq c$,
is a logical explanation
of~$c$ satisfying $\Expl \Rightarrow I$ (hence $\ImpSpace \subseteq \ImpSpace[I]$).
Various interpolation algorithms
exhibit different logical strengths of the produced formulas~\cite{Blicha19}:
Farkas' lemma ($\ItpF$),
its logically stronger variant ($\ItpDF$),
their dual versions ($\ItpF', \ItpDF'$),
and an algorithm with a flexible strength
($\ItpFactor$)~\cite{Alt17}
parametrized by a rational factor~$f \in [0,1]$.
The explanations generated by these algorithms,
when applied to the same arguments,
follow
$\ItpDF \Rightarrow \ItpF \Rightarrow \ItpFactor \Rightarrow \ItpF' \Rightarrow \ItpDF'$.

The next sections
describe the contributions of this paper,
which enable efficient computation of generic logical explanations.
The generalization is not restricted to any particular method
and is suitable for a wide range of solvers.
Sect.~\ref{sec:activation-based}
describes particular techniques
for domain-specific guidance within a parametrized encoding of NNs,
and
Sect.~\ref{sec:guide}
addresses a generic but powerful idea of guiding constraints
that aim at general-purpose logical solvers.

\begin{figure}[t!]
    \centering
    \hspace*{-1.8cm}%
    \includegraphics[trim={1cm 0.0cm 0.75cm 0.0cm},clip,width=0.95\linewidth]{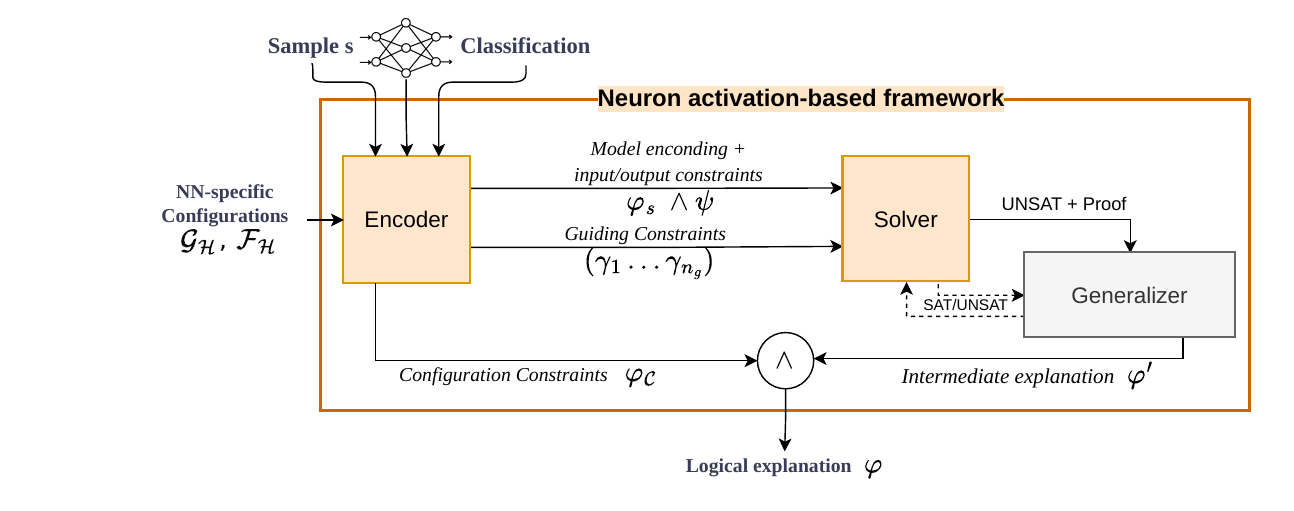}
    \caption{Neuron activation-based framework for automated computation of logical explanations,
    which takes as input a sample point, NN, the classification, and configurations of internal neurons,
    and outputs a logical explanation.
    }
    \label{fig:methods:spexplain}
\end{figure}

\section{Neuron Activation-based Logical Explanations}
\label{sec:activation-based}

This section introduces a flexible symbolic framework
for an efficient, guided computation of logical explanations of the NN behavior,
parametrized by the activations of internal neurons,
and
using logical solvers such as SMT solvers.
In particular,
it focuses on
a comprehensive parametrizable encoder
of the NN classifier---%
namely,
the constraints on weights and biases~(\ref{eq:nn:zed})
and neuron outputs~(\ref{eq:nn:act})---%
which is driven by the activation of internal neurons
at the given sample point being explained,
and by an input configuration specific to the NN.
Notably,
as opposed to many prior approaches,
the encoder is isolated from the symbolic computational engine,
which is addressed in the next section.

Fig.~\ref{fig:methods:spexplain} gives an overview of our framework for computing logical explanations with a parametrized encoding of the input neural network,
highlighting in orange the contributions of this paper.
The \textbf{Encoder}
processes
a given neural network, input sample, and its classification,
yielding a logical formula $\SampleExpl \land \psi$
that is unsatisfiable by construction (cf. Sect.~\ref{sec:encoding}).
Notably,
the \textbf{Encoder} is parametrizable by
NN-specific configurations~\GuideConf{} and~\FixConf{}
that
define specific internal neurons whose activation should be guided (\GuideConf{}) or fixed (\FixConf{}),
optionally producing additional
constraints:
\GuideConf{} controls solely \emph{guiding constraints} $(\GuidePhi[1],\dots,\GuidePhi[\GuideN])$,
while
\FixConf{} may restrict the encoding~$\psi$
as well as produce additional \emph{configuration constraints}~\ConfigPhi{}
to maintain the soundness of the resulting explanation~\Expl{}.
This parametrized encoding helps the logical \textbf{Solver},
such as a proof-producing SMT solver,
to prove the unsatisfiability of $\SampleExpl \land \psi$ faster:
\NEW{%
one of the branches in each \Relu{} activation defined in (\ref{eq:nn:act}) within the encoding~$\psi$
may be firmly pruned (via \FixConf{})
or preferred over the other (via \GuideConf{}).
This becomes especially important
for deep NNs with many hidden layers
(i.e., with many consecutive \Relu{} case splits).
To allow the guidance of the reasoning over $\psi$ via \GuideConf{},
}%
the \textbf{Solver}
implements a new concept of \emph{decision guidance} (see Sect.~\ref{sec:guide}),
given the additional ordered input $(\GuidePhi[1],\dots,\GuidePhi[\GuideN])$
generated based on~\GuideConf{},
where each constraint \GuidePhi[i]
(here an activation of an internal neuron)
is favored to be satisfied rather than not.
Then,
the \textbf{Solver}
extracts
a proof of the unsatisfiability
of $\SampleExpl \land \psi$,
which is subsequently applied to
the \textbf{Generalizer}
that generalizes
formula~\SampleExpl{}
into an intermediate explanation~$\Expl'$.
Combined with the logical encoding of NN configuration constraints~\ConfigPhi{},
the final explanation
results in
$\Expl = \Expl' \land \ConfigPhi$.
The \textbf{Generalizer}
can be \Hide{straightforwardly }implemented in various ways as long as the overapproximation of formula~\SampleExpl{} is computed and
thus directly produces a~$\Expl'$.
For example,
Craig interpolation procedures can be used to extract~$\Expl'$ as a computed interpolant.
Another option is extracting an unsatisfiable core---%
removing parts of $\SampleExpl$ while maintaining its unsatisfiability---%
but requiring additional calls (SAT/UNSAT) to the \textbf{Solver}.
The resulting logical explanations are amenable to further interpretations
(e.g., \cite{LabbafKBFWS:25,Leopardi26}).

The rest of the section describes how the \textbf{Encoder}
forms its outputs
(\textbf{Solver} inputs,
cf. Sect.~\ref{sec:guide}),
given NN configurations~\GuideConf{} and~\FixConf{}:
model encoding~\ClassifierPhi{}
(a part of~$\psi$),
guiding constraints $(\GuidePhi[1],\dots,\GuidePhi[\GuideN])$,
and configuration constraints~\ConfigPhi{}.

\Paragraph{Decision Guidance via Neuron Activations}
\label{sec:guide:nn}

\newcommand*{\GuidePhis}{\ensuremath{\mathbf{\Gamma}}}

\begin{algorithm}[t!]
    \caption{Generating neuron activation-based guiding constraints}
    \label{Alg:Encode:Guide}
    \DontPrintSemicolon
    \KwIn{$\Classifier = \ClassifierTuple$, sample~point~$\Sample \in \FSpace$,
    set~of~neurons~to~guide~$\GuideConf \subseteq \HiddenNSet$}
    \KwOut{ordered set of \GuideN{} guiding constraints $(\GuidePhi[1], \GuidePhi[2], \dots, \GuidePhi[\GuideN])$,
    $\GuideN = |\GuideConf|$}
    \BlankLine

    \Comment{following Sect.~\ref{sec:background:classification},
    each \zed{} refers to~(\ref{eq:nn:zed}),
    \neuron{} to a neuron,}
    \Comment{\dots
    $L$ to the number of layers,
    and
    $n^{(\ell)}$ to the number of neurons in layer~$\ell$}

    $\GuidePhis \gets ()$
    \Comment*{initialize the list of guiding constraints}

    \For(\Comment*[f]{iterate over all hidden layers}){$\ell \gets 1;\ \ell \leq L-1;\ \ell \gets \ell + 1$}{
        \For(\Comment*[f]{iterate over all neurons of layer~$\ell$}){$i \gets 1;\ i \leq n^{(\ell)};\ i \gets i + 1$}{
            \Comment{guide solver to prefer the (in)active phase for (in)active neurons from \GuideConf{}}
            \lIf{$\neuron  \in \GuideConf$}{%
                $\GuidePhis \gets \GuidePhis \Concat \big(\zed(\Sample) > 0 \big)  \mathbin{?}  \big(\zed > 0\big) \mathbin{:}  \big(\zed \leq 0\big)$
            }

        }
    }

    \Return \GuidePhis{}
\end{algorithm}

Logical solvers can handle focused, structure-aware guiding constraints
(cf. Sect.~\ref{sec:guide}),
which we exploit to support direct NN encoding
(cf. Sect.~\ref{sec:encoding})
by suggesting preferred internal neuron activations.
In Fig.~\ref{fig:methods:spexplain},
direct NN encoding corresponds to model encoding~$\psi$,
sent alongside guiding constraints
$(\GuidePhi[1],\dots,\GuidePhi[\GuideN])$
to the \textbf{Solver}.
Alg.~\ref{Alg:Encode:Guide} describes
how to encode neural network~\Classifier{} and a sample point $\Sample \in \FSpace$
into guiding constraints,
driven by a set~\GuideConf{}
(a part of NN configurations in Fig.~\ref{fig:methods:spexplain})
specifying which neurons to guide.
The choice of the set \GuideConf{} is specific to the network~\Classifier{}
and
can be user-provided
or automatically generated
by external NN analysis algorithms.
The algorithm strictly follows the order of hidden layers and neurons
(though other orders are possible).
The guiding constraints focus on the pre-activation~(\ref{eq:nn:zed}) of neurons~\zed{}
evaluated at~\Sample{},
\emph{guiding} the solver toward the desired activation~(\ref{eq:nn:act}),
either $\act = \zed$ or $\act = 0$,
within the ITE encoding.
\NEW{%
For example,
consider a NN with $\FVec{} = (x_1,x_2,x_3)$
and
a neuron $\neuron[1] \in \GuideConf$ (in the first hidden layer)
with pre-activation
$\zed[1] = 2\act[0][1] + 0\act[0][2] + \act[0][3] = 2x_1 + x_3$.
Given sample $\Sample_1 = (1,1,3)$,
Alg.~\ref{Alg:Encode:Guide} adds
guiding constraint
$\zed[1] > 0$
to~\GuidePhis{};
given sample $\Sample_2 = (0,2,0)$,
it adds $\zed[1] \leq 0$.
\ifSpringer
\else
More detailed example follows in Appendix~\ref{apx:example}.
\fi
}%
Notably,
guiding constraints affect neither satisfiability soundness
nor explanation correctness.
Moreover,
the order of neurons and of their activations
is not strictly enforced,
as the solver considers them solely during the decisions,
ignoring them during conflict analysis and propagation.
See Sect.~\ref{sec:guide} for more details.

A direct strategy for Alg.~\ref{Alg:Encode:Guide} is to set $\GuideConf = \HiddenNSet$,
guiding the decisions on all neuron activations in the network.
This is in fact usually the optimal choice
since using guiding constraints
improves both performance and generality in practice.
%

\Paragraph{NN Encoding Parametrizable by Fixing Neuron Activations}
\label{sec:fix:nn}

\newcommand*{\UnstableEps}{\ensuremath{\varepsilon}}

\begin{algorithm}[t!]
    \caption{Parametrizable NN encoding, fixing neuron activations}
    \label{Alg:Encode:Fix}
    \DontPrintSemicolon
    \KwIn{$\Classifier = \ClassifierTuple$, sample~point~$\Sample \in \FSpace$,
    set of neurons~to~fix~$\FixConf \subseteq \HiddenNSet$\Hide{, $X\!=\!Y\!=\!\varnothing$}}
    \KwOut{NN~encoding~\ClassifierPhi{}, configuration~constraints~\ConfigPhi{}}
    \BlankLine

    \Comment{notation as in Alg.~\ref{Alg:Encode:Guide};
    \act{} refers to~(\ref{eq:nn:act});
    $\act[0] = x_i$ (input feature variables)}

    $\ClassifierPhi \gets \top, \ConfigPhi \gets \top$
    \Comment*{initialize encodings of $\Classifier$ and the configuration constraints}

    \For(\Comment*[f]{iterate over all hidden layers}){$\ell \gets 1;\ \ell \leq L-1;\ \ell \gets \ell + 1$}{
        \For(\Comment*[f]{iterate over all neurons of layer~$\ell$}){$i \gets 1;\ i \leq n^{(\ell)};\ i \gets i + 1$}{
            \uIf(\Comment*[f]{do not fix neurons that are not included in \FixConf{}}){$\neuron \not \in \FixConf$}{
                \label{Alg:Encode:Fix:Cond}
                $\ClassifierPhi \gets \ClassifierPhi \land \Big( \act = \mathit{ite}\big(\zed > 0, \zed, 0\big) \Big)$
                \Comment*{direct ITE encoding}
            }
            \Else{
                $\ClassifierPhi \gets \ClassifierPhi \land \big( \zed(\Sample) > 0 \big) \mathbin{?} \big( \act = \zed\big) \mathbin{:} \big( \act = 0\big)$
                \Comment*{fix the activation}
                $\ConfigPhi \gets \ConfigPhi \land \big( \zed(\Sample) > 0 \big) \mathbin{?} \big( \zed > 0\big) \mathbin{:} \big( \zed \leq 0\big)$
                \Comment*{constrain the activation}
            }
        }
    }




    $\ClassifierPhi \gets \ClassifierPhi \land \bigwedge_{i=1}^{n^{(L)}} (\act[L] = \zed[L])$
    \Comment*{no \Relu{} in the output layer}

    \Return \ClassifierPhi{}, \ConfigPhi{}
\end{algorithm}

This part describes how to encode
a neural network~\Classifier{}
into formula~\ClassifierPhi{}
(cf. Sect.~\ref{sec:encoding})
while
optionally fixing selected neurons to one activation phase---%
that is,
selecting one branch of the \Relu{} splits (\ref{eq:nn:act}).
In Fig.~\ref{fig:methods:spexplain},
\ClassifierPhi{} is a conjunct of the model encoding~$\psi$.
If neurons are fixed,
\ClassifierPhi{} encodes a stricter variant of~\Classifier{}.
To ensure that the resulting explanation~\Expl{} is valid w.r.t.~\Classifier{},
the \textbf{Encoder} also produces configuration constraints~\ConfigPhi{}:
in Fig.~\ref{fig:methods:spexplain},
these are applied to the intermediate explanation~$\Expl'$
produced by the \textbf{Generalizer}
such that $\Expl := \Expl' \land \ConfigPhi{}$.
Then,
\Expl{} constrains parts of the feature space
where only the fixed activation phase of neurons
is reachable.
This technique is \emph{orthogonal} to the decision guidance above,
which addresses only optional neuron activation guiding constraints.
Alg.~\ref{Alg:Encode:Fix} describes
the process of encoding~\ClassifierPhi{}
and configuration constraints~\ConfigPhi{},
given \Classifier{}, a sample point~\Sample{},
and a set of neurons~\FixConf{} to be fixed.
The set \FixConf{} is an input similar to \GuideConf{} in Alg.~\ref{Alg:Encode:Guide},
but forms a \emph{separate} part of the NN configurations in Fig.~\ref{fig:methods:spexplain}.
Alg.~\ref{Alg:Encode:Fix} encodes the constraints
depending on the input feature variables $\FVec{} = (x_1,\dots,x_m)$
and again strictly follows the order
of layers and neurons.
Direct encoding using an ITE term is applied
for neurons not specified in~\FixConf{};
otherwise,
only one of the branches is encoded,
depending on the evaluation of the pre-activation value~\zed{}
at the sample point~\Sample.
The algorithm also stores the pre-activation constraint
into configuration constraints~\ConfigPhi{}\Hide{\footnote{
    It would not be sufficient to only encode the pre-activation constraint
    as a part of~\ClassifierPhi{},
    because later the proof generalizer would still be allowed to relax this constraint,
    while it must hold firmly.
}}
to maintain the restriction
and soundness of the resulting explanations.
\NEW{%
For example,
consider the same NN
and
neuron~\neuron[1].
If $\neuron[1] \notin \FixConf$,
its activation is encoded as
$\act[1] = \mathit{ite}(\zed[1] > 0, \zed[1], 0)$
as part of \ClassifierPhi{}.
If $\neuron[1] \in \FixConf$,
given sample $\Sample_1 = (1,1,3)$,
$\act[1] = \zed[1]$;
given sample $\Sample_2 = (0,2,0)$,
$\act[1] = 0$,
while adding $\zed[1] > 0$
and
$\zed[1] \leq 0$
to \ConfigPhi{},
respectively.
\ifSpringer
\else
More detailed example follows in Appendix~\ref{apx:example}.
\fi
}%
At the end of Alg.~\ref{Alg:Encode:Fix},
it encodes output neuron constraints
where no \Relu{} function is applied.
Notably,
such a \Expl{} is a correct logical explanation w.r.t.~\Classifier{}
(Definition~\ref{def:spaceex})
since the form of the formula is not restricted.

\begin{theorem}
\label{th:FixExpl}
    Given a neural network~\Classifier{}
    and
    a sample point $\Sample \in \FSpace$
    classified as class~$c$
    by \Classifier{},
    let
    $\ClassifierPhi'$ and \ConfigPhi{} be the formulas produced by Alg.~\ref{Alg:Encode:Fix},
    $\Classifier'$ be the restricted NN encoded by $\ClassifierPhi'$,
    and
    $\Expl'$ be a logical explanation of~$\Classifier'$ and~$c$
    such that
    $\SampleExpl \implies \Expl'$.
    Then,
    $\Expl := \Expl' \land \ConfigPhi$
    is a logical explanation of~\Classifier{} and class~$c$.
\end{theorem}

\newcommand*{\ProofFixExpl}{%
\begin{proof}
    Let
    \ClassifierPhi{} be the encoding of the original network~\Classifier{},
    then
    $\psi := \ClassifierPhi \land \DomainPhi \land \neg \ClPhi$
    encodes the change of classification~$c$ in~\Classifier{}.
    We need to prove that
    \Expl{} is satisfiable
    and
    $\Expl \land \psi$ is unsatisfiable.
    The satisfiability of
    \Expl{} (i.e., $\Expl' \land \ConfigPhi$)
    follows since
    $\SampleExpl \land \ConfigPhi$
    is satisfiable
    (\ConfigPhi{} is constructed based on \Sample{})
    and the premise
    $\SampleExpl \Rightarrow \Expl'$.
    Let
    $\psi' := \ClassifierPhi' \land \DomainPhi \land \neg \ClPhi$.
    Since
    $\Expl' \land \psi'$ is unsatisfiable,
    so is $\Expl' \land \psi \land \ConfigPhi$
    because
    $\ClassifierPhi \land \ConfigPhi \equiv \ClassifierPhi'$
    and hence
    $\psi \land \ConfigPhi \equiv \psi'$.
    Therefore,
    $\Expl \land \psi$ is also unsatisfiable.
    \QED
\end{proof}%
}
\ProofFixExpl{}
This universally
produces logical explanations from sample points.
For instance,
given \Classifier{}, sample~\Sample{} classified as~$c$,
\ClassifierPhi{} and \ConfigPhi{} produced by Alg.~\ref{Alg:Encode:Fix},
and
using an interpolant
$\Expl' = \Itp(\SampleExpl, \ClassifierPhi \land \DomainPhi \land \neg \ClPhi)$,
where \Itp{} is a Craig interpolator,
then
$\Expl := \Expl' \land \ConfigPhi$ is a logical explanation of~\Classifier{} and class~$c$
such that
$\SampleExpl \Rightarrow \Expl$.

\NEW{%
When
the activation phase of a neuron at sample~\Sample{}
is \emph{unstable}---%
i.e.,
pre-activation $\zed(\Sample)$ in Alg.~\ref{Alg:Encode:Fix} is close to zero---%
fixing such a neuron may be too restrictive,
potentially resulting in a non-robust explanation.
This can be avoided
by extending the condition
in line~\ref{Alg:Encode:Fix:Cond}
with a disjunction for
$|\zed(\Sample)| \leq \UnstableEps$,
using a small~\UnstableEps{}
(e.g., $\UnstableEps = 10^{-3}$
for normalized domains \Domain[i])
to enforce direct encoding.
}

A straightforward strategy is to set
$\GuideConf = \emptyset$ for Alg.~\ref{Alg:Encode:Guide}
and
$\FixConf = \HiddenNSet$ for Alg.~\ref{Alg:Encode:Fix},
which removes all branchings from \ClassifierPhi{}
(if no activations are unstable).
Even then,
$\psi$ may not be deterministic,
due to classification constraints~\ClPhi{}
(cf. Sect.~\ref{sec:encoding}).
The combination of Alg.~\ref{Alg:Encode:Fix} and Alg.~\ref{Alg:Encode:Guide}
can be optimized
not to produce guiding constraints for neurons that have been fixed.
Hence,
it still makes sense to apply the input sets with
$\FixConf{} \cap \GuideConf{} \neq \emptyset$
(neurons with unstable activation are not fixed),
and
a meaningful combination is even
$\FixConf = \GuideConf = \HiddenNSet$.
However,
if fixing all neurons is overly restrictive,
it is useful to identify~\FixConf{} as a subset of~\HiddenNSet{}.
%
The following paragraph introduces an algorithm to identify relevant neurons,
which is useful for finding an efficient combination
of configurations~\FixConf{} and~\GuideConf{}.

\Paragraph{Finding Relevant Neurons by Crossing the Decision Boundary}


\Hide{
{~}
\gri{For the compression reasons, I suggest to replace the pseudocode and the original text (in brown, below) with the blue text: the formula captures exactly the output of the algorithm}
{\color{brown}
Alg.~\ref{Alg:interesting_neurons} describes the identification
of relevant neurons in hidden layers.
Given a NN and two sample points~\Sample{} and~$\Sample'$
that are classified into different classes
by the network, it returns a set of \emph{relevant neurons}~\IntConf{},
where the activation phase
evaluated at the two samples differs.}
}

Relevant neurons $\IntConf$ in hidden layers can be identified
given two sample~points $\Sample, \Sample' \in \FSpace$
classified into different classes by the network
(i.e., $\ClFun(\Sample) \neq \ClFun(\Sample')$):
\begin{equation}
\label{eq:interesting_neurons}
\IntConf = \big\{ \neuron \mid \ell \in [1, L), i \in [1, n^{(\ell)}], (\act(\Sample)=0) \neq (\act(\Sample')=0 )\big\}
.
\end{equation}
\NEW{The \textbf{Encoder} leaves these neurons \emph{unfixed}.
By~(\ref{eq:interesting_neurons}),
their activation phase changes when evaluated at~\Sample{} and $\Sample'$
across the decision boundary.
This preserves the flexibility needed to capture feature relations associated with this boundary crossing.}
All other neurons
are fixed
($\FixConf = \HiddenNSet \setminus \IntConf$),
with
$\GuideConf = \HiddenNSet$ as the fallback.

\Hide{
\begin{algorithm}[t!]
    \caption{Identifying relevant neurons by crossing the boundaries}
    \label{Alg:interesting_neurons}
    \DontPrintSemicolon
    \KwIn{$\Classifier = \ClassifierTuple$, sample~points ~$\Sample, \Sample' \in \FSpace$ such that $\ClFun(\Sample) \neq \ClFun(\Sample')$}
    \KwOut{set~of~relevant~neurons~$\IntConf \subseteq \HiddenNSet$}
    \BlankLine
    \Comment{notation as in Alg.~\ref{Alg:Encode:Fix}}
    \BlankLine


    $\IntConf \gets \emptyset$

    \For(\Comment*[f]{iterate over all hidden layers}){$\ell \gets 1;\ \ell \leq L-1;\ \ell \gets \ell + 1$}{
        \For(\Comment*[f]{iterate over all neurons of layer~$\ell$}){$i \gets 1;\ i \leq n^{(\ell)};\ i \gets i + 1$}{

            \If(\Comment*[f]{check if activation phases differ}){$ (\act(\Sample)=0) \neq (\act(\Sample')=0 )$ }{
                            $\IntConf \gets \IntConf \cup \{\neuron\}$
            }

        }
    }
    \BlankLine
    \Return \IntConf
\end{algorithm}
}

\NEW{%
To restrict the state space and maximize the benefits of fixing,
\IntConf{} must be kept small.
For this reason,
}%
$\Sample'$ is computed as an \emph{adversarial example}~\cite{DBLP:journals/corr/GoodfellowSS14}
of sample~\Sample{},
by adding a small,
carefully chosen perturbation to~\Sample{}
that changes the classification.
In the literature, adversarial examples are typically constructed by solving a constrained optimization problem
over a small perturbation
to maximize the loss of the original class, often via iterative search within a neighborhood of the original sample~\Sample{}.
A popular technique for this purpose is \emph{Projected Gradient Descent} (PGD)~\cite{DBLP:conf/iclr/MadryMSTV18}.
Starting from an initial point close to~\Sample{},
using a fixed number of iterations,
PGD updates the position
in a direction that increases the loss of the classifier.
After each update, the point is projected back
into a small vicinity of~\Sample{}, bounding the perturbation
and the resulting adversarial example.


Overall,
each presented technique contributes to the performance
of computing NN explanations
and can be combined on demand,
as illustrated in Sect.~\ref{sec:exp}.

\section{Decision Guidance of Logical Solvers}
\label{sec:guide}


\SetKwFunction{CheckSat}{CheckSat*}
\SetKwFunction{SelectDecisionLiteral}{SelectDecisionLiteral*}
\SetKwFunction{SelectInternalDecisionLiteral}{SelectInternalDecisionLiteral}
\SetKwFunction{MakeNewVariables}{MakeNewVariables}
\SetKwFunction{MakeLiteral}{MakeLiteral}

\newcommand*{\Assignment}{\ensuremath{\mathbf{A}}}

\begin{algorithm}[t!]
    \caption{Check the satisfiability with propositional decision guidance}
    \label{Alg:CheckSat:Guide}
    \DontPrintSemicolon
    \KwIn{formula~\Formula, finite ordered set of \GuideN{} guiding constraints $(\GuidePhi[1], \GuidePhi[2], \dots, \GuidePhi[\GuideN])$}
    \KwOut{\Sat{} or \Unsat{} if \Formula{} is satisfiable or not, resp.}
    \BlankLine

    \ColorEmph{
    $(\GuideVar[1], \GuideVar[2], \dots, \GuideVar[\GuideN]) \gets \MakeNewVariables{\GuideN}$
    \Comment*{\GuideN{} fresh propositional variables}
    \label{Alg:CheckSat:Guide:Init:Begin}
    $\Formula \gets \Formula \land \bigwedge_{i=1}^{\GuideN} \big( \GuideVar[i] \implies \GuidePhi[i] \big)$
    \Comment*{update the formula}
    \label{Alg:CheckSat:Guide:Init:End}
    }

    \Comment{get a propositional assignment \Assignment}
    $\Assignment \gets \CheckSat{\Formula, (\GuideVar[1], \GuideVar[2], \dots, \GuideVar[\GuideN])}$
    \Comment*{calls \SelectDecisionLiteral inside}
    \label{Alg:CheckSat:Guide:CheckSat}


    \If(\Comment*[f]{full assignment}){\Assignment{} satisfies \Formula{} and contains all propositional variables}{
        \Return \Sat
    }

    \Return \Unsat

    \BlankLine

    \Function{\SelectDecisionLiteral{$\Assignment, (\GuideVar[1], \GuideVar[2], \dots, \GuideVar[\GuideN])$}}{
        \ColorEmph{
        \For(\Comment*[f]{strictly follow the order}){$i \gets 1;\ i \leq \GuideN;\ i \gets i + 1$}{
            \label{Alg:CheckSat:Guide:SelectUserLit:Begin}
            \lIf(\Comment*[f]{select literal \GuideVar[i] if not assigned yet}){$\GuideVar[i] \not \in \Assignment$}{%
                \Return \GuideVar[i]
            }
        }
        \label{Alg:CheckSat:Guide:SelectUserLit:End}
        }

        \BlankLine

        \Return \SelectInternalDecisionLiteral{\Assignment{}}
        \Comment*{use internal heuristics}
    }
\end{algorithm}

Logic-based explainability methods rely on efficient symbolic logical solvers,
such as SAT, SMT, or MILP,
as underlying computational engines.
These solvers are highly optimized for general constraint solving.
However, generic heuristics may not be optimal for specific domain-driven problems.
This paper proposes guiding constraints as a transferable mechanism to optimize the performance of logical solvers
for NN analysis.
A \emph{guiding constraint} is a formula
defining a preferred property rather than a mandatory requirement,
which can be a natural, yet optional,
part of the input.
The solver implementation
is still allowed to completely ignore it.
Given a formula~\Formula{} and a guiding constraint \GuidePhi{},
we introduce a \emph{fresh} propositional variable~\GuideVar{},
called the \emph{guiding variable} of \GuidePhi{}.
Notably,
\GuidePhi{} is not limited to existing variables,
literals, or even terms in \Formula{}.
Then, because \Formula{} is equisatisfiable with
$\Formula \land (\GuideVar \Rightarrow \GuidePhi)$,
a solver can enforce \GuidePhi{}
via a propositional decision on~\GuideVar{}.
Internally,
the solver constructs a sequence of guiding variables~\GuideVar[i]
that serve as systematic decision suggestions
for logical constraints~\GuidePhi[i].
%
Crucially,
the \emph{order} of decision suggestions \emph{matters}:
many models naturally follow a domain-specific topology
that is entirely lost in
standard logical encodings focused solely on firm constraints.
Moreover,
extracting an effective decision order
from structured models,
such as a NN aligned into layers,
is straightforward.

Alg.~\ref{Alg:CheckSat:Guide} describes the main idea
of the decision guidance
via propositional decisions
within a logical solver checking the satisfiability of a formula~\Formula{},
given an \emph{ordered} set of guiding constraints.
Although the algorithm is straightforward
and
introduces only a few solver modifications
(highlighted in blue),
implementation requires internal knowledge of the target tool
to integrate the feature into the decision heuristics,
often within a highly optimized environment.
In lines~\ref{Alg:CheckSat:Guide:Init:Begin}--\ref{Alg:CheckSat:Guide:Init:End},
the algorithm creates guiding variables and updates formula \Formula{}.
In line~\ref{Alg:CheckSat:Guide:CheckSat},
it calls \CheckSat,
which differs from standard solving
by using \SelectDecisionLiteral
instead of \SelectInternalDecisionLiteral
for decision choices of propositional literals.
In \SelectDecisionLiteral
(lines~\ref{Alg:CheckSat:Guide:SelectUserLit:Begin}--\ref{Alg:CheckSat:Guide:SelectUserLit:End}),
the solver first attempts to pick an unassigned guiding variable,
selecting the positive literal~\GuideVar[i]
as the decision literal.
As long as some~\GuideVar[i] remains unassigned,
decision guidance takes precedence
over default internal heuristics.
Notably,
\SelectDecisionLiteral affects only decision choices,
leaving conflict reasoning and propagation untouched.
Hence,
any \GuideVar[i] that conflicts with the current assignment
is skipped
($\GuideVar[i] \not \in \Assignment$ in line~\ref{Alg:CheckSat:Guide:SelectUserLit:End}).

When \GuidePhi[i] is itself a propositional literal
or has an existing propositional abstraction in the solver,
the implementation can omit creating \GuideVar[i]
and use that literal directly.
While Alg.~\ref{Alg:CheckSat:Guide} targets
DPLL-based SMT solvers,
frameworks supporting non-propositional decisions
(e.g., MCSat~\cite{DeMouraJ:13})
can implement decision guidance without explicit guiding variables.

A restricted form of this capability
was previously introduced in \Mathsat~\cite{mathsat5},
requiring guiding constraints to be propositional literals
in the input formula~\Formula{}.
The novel, unrestricted feature
was first implemented
in \Zthree~\cite{z3}
and \Opensmt{},
where
guiding the computation of NN explanations
yielded orders-of-magnitude speedups over unguided reasoning---%
a gain transferable to any supporting SMT solver.

This general concept applies to diverse logical solvers
and domains.
In this paper,
we leverage decision suggestions
to guide the NN explanation process
toward specific neuron activations,
aligned with a sample point
and
the topological order of a network.
Concretely,
this instantiation appears in
Fig.~\ref{fig:methods:spexplain}
within the \textbf{Solver}
implementing Alg.~\ref{Alg:CheckSat:Guide},
accepting guiding constraints $(\GuidePhi[1],\dots,\GuidePhi[\GuideN])$
from the \textbf{Encoder}
as input.
This application utilizes a subset of the full capability:
individual constraints~\GuidePhi[i]
(neuron activations)
are already present in \ClassifierPhi{},
though Alg.~\ref{Alg:CheckSat:Guide}
supports arbitrary fresh constraints.

\section{Experimental Evaluation}
\label{sec:exp}

We implemented the neuron activation-based approach
using
a parametrized encoding of the NN w.r.t. a specific configuration of the internal neurons
within the \textbf{Encoder}
(Alg.~\ref{Alg:Encode:Guide} and~\ref{Alg:Encode:Fix})\footnote{
    \url{https://github.com/usi-verification-and-security/spexplain}
}.
Computation of configurations that are based on~(\ref{eq:interesting_neurons}) is implemented
using TorchAttack library~\cite{kim2020torchattacks} for PGD.
We stress-tested the performance of our techniques on the classical image-classification datasets,
traditionally used to study the performance limits of logical explanation techniques for classification:
GTSRB (3072 inputs, 43 classes),
CIFAR10~\cite{krizhevsky2009learning} (3072 inputs, 10 classes),
and MNIST (784 inputs, 10 classes).
We also used two smaller tabular medical domain benchmarks
that operate on very concrete domain inputs
and are used by medical practitioners:
heart-attack risk~\cite{heart_disease_45} (13 inputs, 2 classes),
and obesity~\cite{obesity} (15 inputs, 7 classes) datasets.
Here,
not only the individual features
but also the relationships among them are highly relevant.
For each dataset,
we randomly selected 100 sample points,
and
trained fully connected NNs
with 1 to 14 layers
for the image-classification datasets
and
with 1 to 8 layers
for the medical datasets\footnote{
    In the medical datasets,
    the training
    of
    NNs beyond 8 layers
    often does not converge.
},
each
with 50 neurons per hidden layer.
%
This paper focuses primarily on the monolithic approach that is based on SMT,
evaluated in Sect.~\ref{sec:exp:smt}.
Additionally,
in Sect.~\ref{sec:exp:chc},
we evaluate an alternative experimental modular variant to observe
how internal neuron activation configurations influence non-interpolation-based NN analysis.
Although we omit the formal details of the modular encoding for brevity,
we provide the underlying intuition,
noting that specific optimizations from Alg.~\ref{Alg:Encode:Fix} can be effectively migrated to this setting.

\subsection{SMT and Craig interpolation Instantiation}
\label{sec:exp:smt}

\begin{figure}[t!]
    \centering
    \includegraphics[trim={0.35cm 0.4cm 0cm 0.1cm},clip, width=0.94\linewidth]{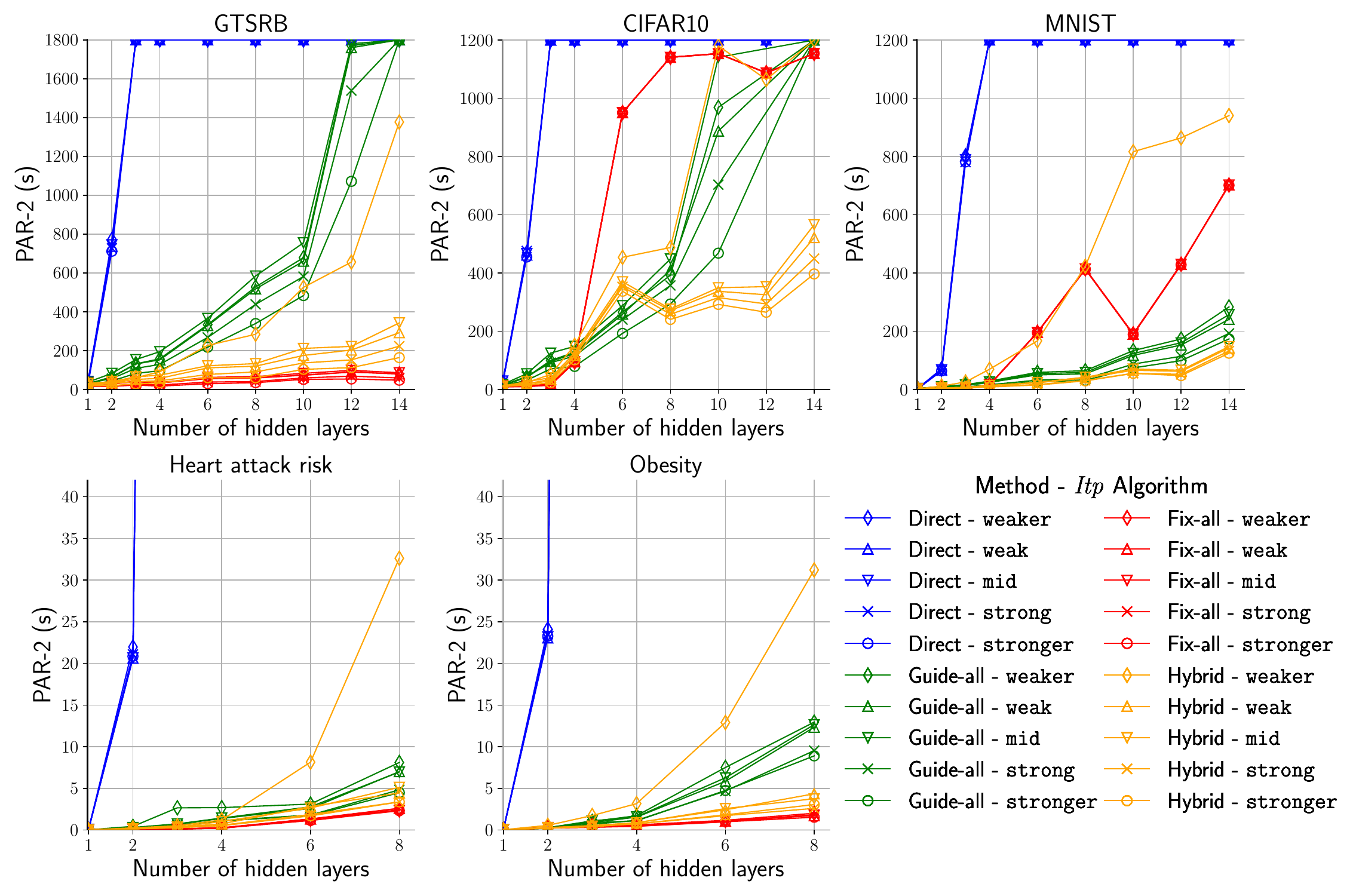}
    \caption{Performance comparison of our methods and interpolation algorithms.
     In medical datasets, it shows PAR-2 only up to 40 to make methods other than \naive{} visible.}
    \label{fig:runtime_compare50}
    \includegraphics[trim={0.35cm 0.4cm 0cm 0.1cm},clip, width=0.94\linewidth]{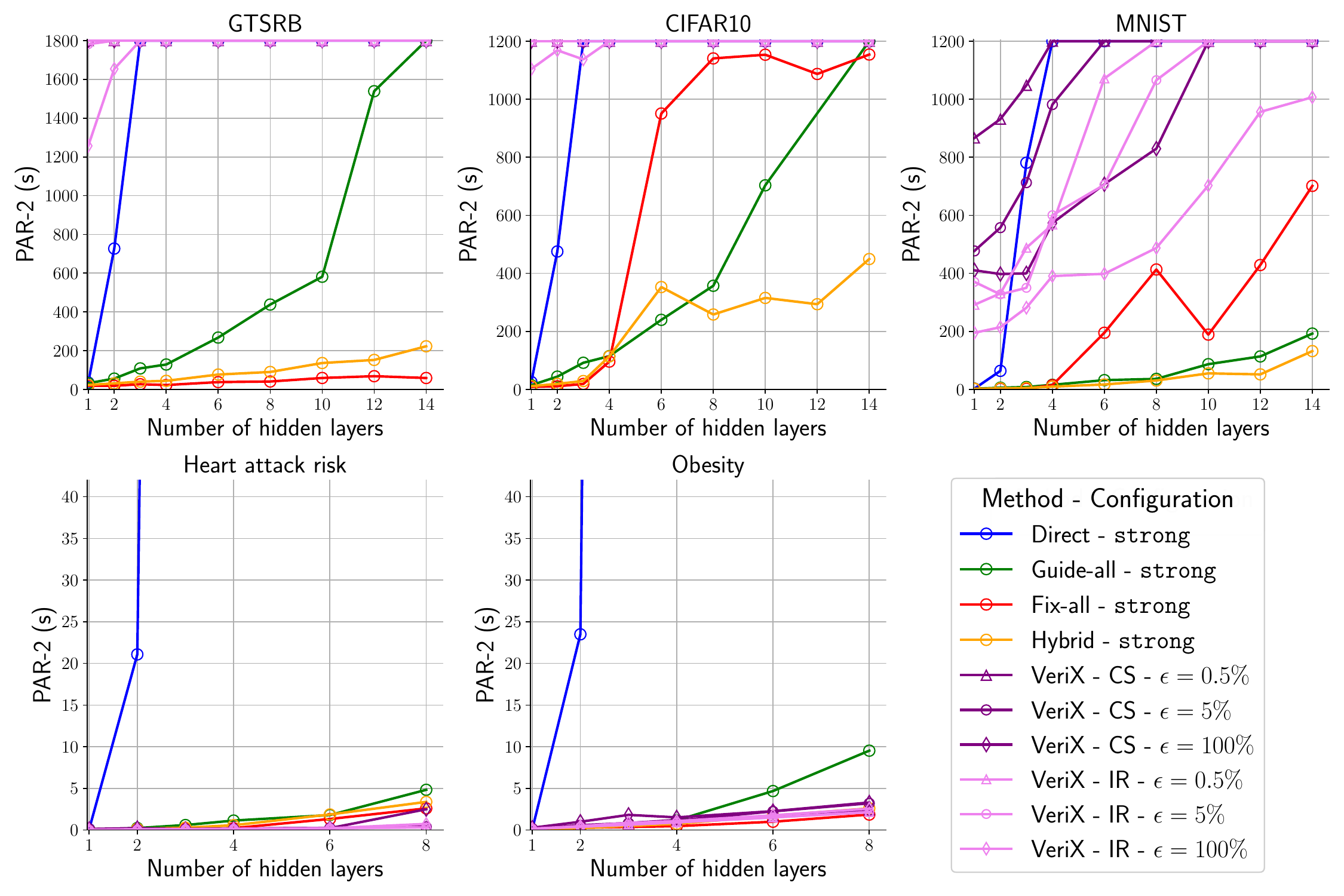}
    \caption{Performance comparison of \astrong{} variants with \Verix{}\Hide{ on networks with 50 neurons per layer}.}
    \label{fig:verix_compare}
\end{figure}

We instantiated the framework from Fig.~\ref{fig:methods:spexplain}
with the \Opensmt{}
as the \textbf{Solver}
and
used its interpolating features
as the procedures for the \textbf{Generalizer}.
To support decision guidance capability,
we implemented Alg.~\ref{Alg:CheckSat:Guide} in \Opensmt{} (\textbf{Solver}).
The choice of tools is intentional
to support a direct comparison with the experimental study~\cite{LabbafKBFWS:25},
referred to as \naive{}.
%
We compare four different configurations of our method:
    \naive: $\GuideConf = \FixConf = \emptyset$
    (the implementation used in~\cite{LabbafKBFWS:25});
    \guideall: $\GuideConf = \HiddenNSet$, $\FixConf = \emptyset$;
    \fixall:  $\GuideConf = \emptyset$, $\FixConf = \HiddenNSet$;
    and
    \hybrid: $\GuideConf = \HiddenNSet$,
    $\FixConf = \HiddenNSet \setminus \IntConf$
    (using~(\ref{eq:interesting_neurons})
    with adversarial example $\Sample'$).
Instances of interpolation algorithms
(cf. Sect.~\ref{sec:background})
use the notation
$\ItpDF \mapsto \astronger$,
$\ItpF \mapsto \astrong$,
$\ItpF' \mapsto \aweak$,
$\ItpDF' \mapsto \aweaker$,
and
$\ItpFactor \mapsto \amid{}$
with
$f := 0.5$.
Depending on the complexity of the dataset,
the evaluation uses the following computation time limits per individual explanation:
GTSRB: 15~min; MNIST and CIFAR10: 10~min; medical datasets: 5~min.
To aggregate both the runtime and occasional timeouts, we leverage the metric commonly used in the SMT community, \partwo{}, which is the runtime, or in the case of a timeout, twice the timeout.
For every NN and explanation technique,
we report the average \partwo{} out of all 100 explanations.
The experiments were run on a Linux 5.4 machine
with 256~GB physical memory
and
AMD\textsuperscript{\textregistered} EPYC 7452 64-threaded CPU.

Fig.~\ref{fig:runtime_compare50} shows how different configurations and interpolation algorithms scale to the number of layers,
running on particular datasets.
The x-axis shows the number of hidden layers, and the y-axis shows the \partwo{} average.
The plots show the clear trend that the original \naive{} method always times out when using models with at most 4 layers.
In contrast, the new techniques scale substantially better.
Notably, even \guideall{} that encodes the entire network still significantly improves over \naive{}.
The \fixall{} configuration dominates the performance in the medical domain tasks and GTSRB, but
experiences timeouts at CIFAR10 and MNIST within deeper networks
due to a weakness of the theory solver in \Opensmt{} (simplex)---%
it struggles at resolving the overly restricted problem,
approaching the worst-case execution.
The \hybrid{} configuration addresses this issue and usually achieves the best or second-best performance.
The choice of interpolation algorithm has low impact on scaling trends,
with one notable exception, \hybrid{} with \aweaker{},
when it produces many non-convex constraints.
Explaining with \astronger{} performs best
but it usually does not generalize \Hide{the sample point }at all.
In the following, we select the second-best performing algorithm \astrong.

We compare to the state-of-the-art NN explanation tool
\Verix{}\footnote{
    \NEW{
    \Verix{}+ could not be successfully executed on our evaluation suite,
    running only on the simpler benchmarks with shallow networks reported in~\cite{Wu_Li_Wu_Barrett_2026}.
    It would be also interesting to compare with the recent results of~\cite{bassan2025explaining},
    which remains future work.
    }
},
using two configurations~\cite{wu2023verixverifiedexplainabilitydeep}:
\emph{complete} verification with \emph{sensitivity} traversal (CS),
which relaxes the sample point most,
and
\emph{incomplete} verification with \emph{random} traversal (IR),
the most performant option.
A further parameter is the perturbation radius $\epsilon$ that defines the maximal vicinity of the sample point.
We adopted~\cite{wu2023verixverifiedexplainabilitydeep}
$\epsilon \in \{ 0.5\%, 5\%, 100\%  \}$.
The value $\epsilon=100\%$ is conceptually similar to our approach, as features are allowed to range in entire domain.
Fig.~\ref{fig:verix_compare} compares the results from Fig.~\ref{fig:runtime_compare50}
(using \astrong{}) and \Verix{}.
\Verix{} performs better using IR configuration or with higher values of~$\epsilon$.
\Verix{} performs significantly worse in the image-classification domain,
because it is sensitive to the number of input features~\cite{LabbafKBFWS:25},
and slightly outperforms all our methods only in the heart-attack risk case.
Recall that here the usefulness of the explanations produced by \Verix{} is limited by not capturing relations between features that are important in tabular data.
\NEW{We further experimented with more complex networks (200 neurons per layer).
These experiments are much more challenging,
with more methods timing out for some network sizes,
though key trends from Fig.~\ref{fig:verix_compare} persist.
For the results and discussion, refer to
\ifSpringer
the extended version of this paper~\cite{extended_version}.
\TodoFL{add arxiv citation.}
\else
Appendix~\ref{apx:performance}.
\fi
}%
\Hide{
Fig.~\ref{fig:verix_compare200} compares the performance of the same methods as in Fig.~\ref{fig:verix_compare}, selecting just the performant option $\epsilon=100\%$ for \Verix{}, on more complex networks with 200 neurons per hidden layer.
While some observations from Fig.~\ref{fig:verix_compare} remain similar, these experiments are much more challenging and only some of the methods manage to compute explanations for all sizes of the networks. Notably, \fixall{} still scales very well for GTSRB.
\Verix{} fails to compute the explanations in most image-classification cases, while it performs well for medical datasets.
}

\begin{figure}[t!]
    \centering
    \includegraphics[trim={0.3cm 0.28cm 2.8cm 1.3cm},clip, width=0.95\linewidth]{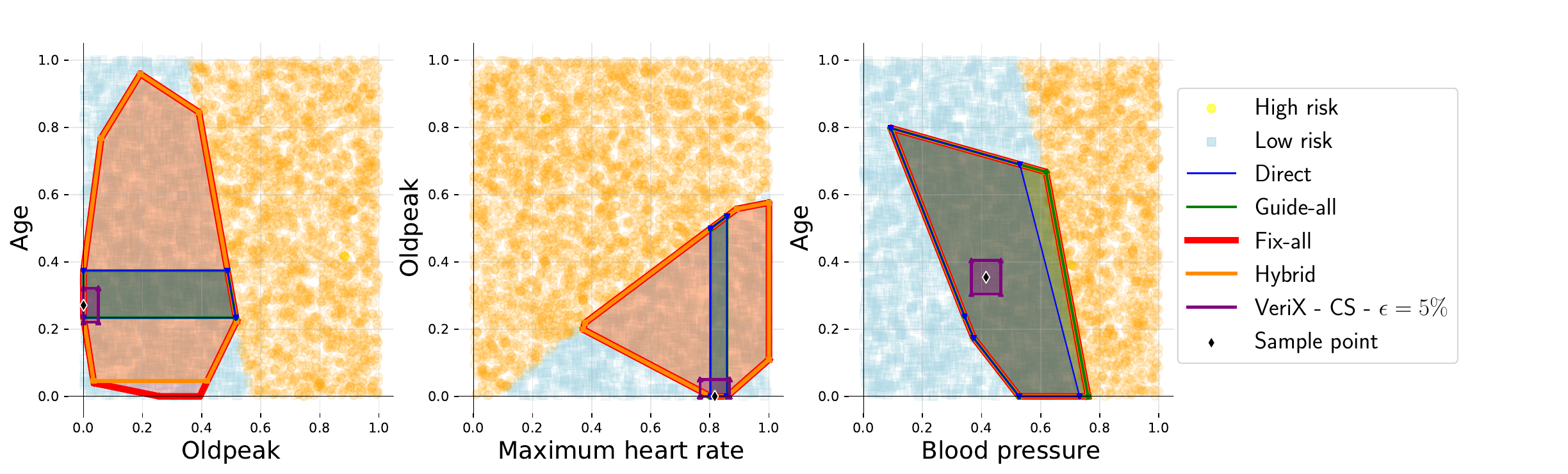}
    \caption{Visual comparison of selected 2D snapshots of explanation spaces
    of the heart-attack dataset (with two classes: high and low risk)
    produced by different methods.}
    \label{fig:class_space}
    \includegraphics[trim={1cm 0.2cm 2cm 0cm},clip,width=0.92\linewidth]{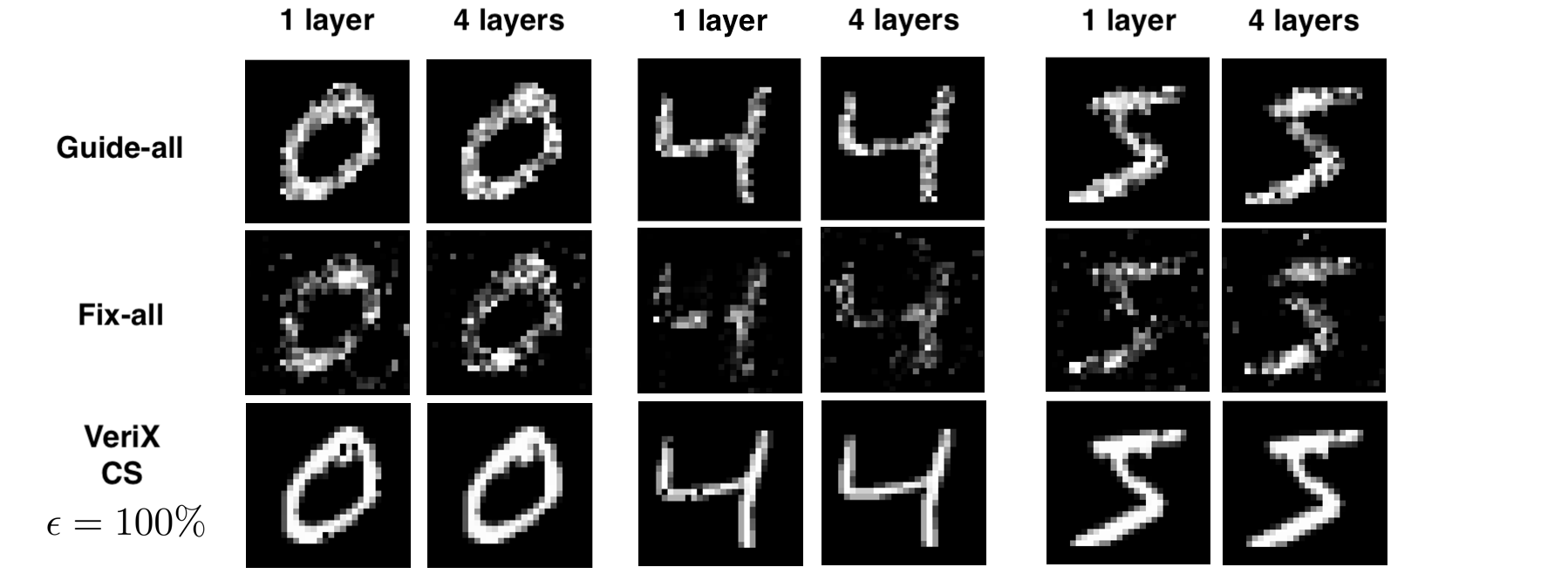}
    \caption{Visual comparison of generality of explanations produced by different techniques from three MNIST samples of digits 0, 4, and 5. 
    The number of pixels that are darker compared to the original image reflects more general explanations.}
    \label{fig:Norm_points_fixall}
\end{figure}

We checked the generality of \astrong{} explanations\footnote{%
    Algorithm \astrong{} focuses on a close neighborhood of the sample point
    but still often covers spaces with a similar shape compared to weaker algorithms~\cite{LabbafKBFWS:25}.
}
computed by our methods
using
subsumption checks for all pairs of explanations
for heart-attack networks having~1 to~4 layers\footnote{%
    Beyond 4 layers, \naive{} never finishes,
    and
    the subset checks are also very expensive.
}.
There is a clear trend in the generality:
98\% of explanations produced by \naive{}
and \guideall{}
are equivalent,
while
100\% and 95\% of \guideall{} explanations
are subsumed by \fixall{} and \hybrid{} explanations,
respectively,
and
98\% of \hybrid{} explanations are subsumed by \fixall{}.
If using weaker interpolation algorithms,
the checks
frequently result in inconclusive cases
(explanations intersect but do not subsume each other)
due to the higher flexibility
in the generalization.
While the dominance of \fixall{} may be counter-intuitive,
it reveals that
\begin{enumerate*}[label=(\arabic*)]
    \item the continuous area restricted by the neuron activations
    is still large enough for the generalizations,
    and
    \item since this area is defined by more deterministic,
    convex constraints,
    it is easier for the interpolation to generalize the constraints.
\end{enumerate*}
These results suggest
that the generality of explanations is not compromised
or even improves
when restricting the NN.

\NEW{
The following
further supports observations
on generality
and demonstrates that the explanations remain useful and interpretable.
}%
Fig.~\ref{fig:class_space}
visually compares
selected two-dimensional snapshots of explanations in the feature space~\cite{LabbafKBFWS:25},
using
the single-layer heart-attack network
and
more general algorithm \aweak{}.
Explanations computed by \fixall{} and \hybrid{}
again subsume other explanations,
providing the highest generality and also most precise approximation of the decision boundary.
\naive{} and \guideall{} are also general,
not restricted to a specific vicinity of the sample points.
By contrast,
the explanations computed by \Verix{}
(CS, $\epsilon = 5\%$)
represent only small rectangles in the feature space.

In image-recognition domain,
Fig.~\ref{fig:Norm_points_fixall} demonstrates the generality
of \astrong{} explanations for the MNIST dataset,
showing specific points from the feature space (i.e., digit images)
that are different from the original sample point
but still
covered by the explanation.
The points are computed using~\cite{Leopardi26},
the euclidean minimum norm,
and illustrate
how much each pixel of the sample
can be darkened while still guaranteeing the classification,
capturing the flexibility given by the explanation
and hence reflecting the size of its impact space.
It compares \guideall{}, \fixall{},
and \Verix{} (CS, $\epsilon = 100\%$),
confirming that \fixall{} yields more general explanations than \guideall{}.
Despite the use of \astrong{} algorithm,
both are still substantially more general than \Verix{}
where only a few isolated pixels are switched to black.
We also observed that \Verix{} (CS, $\epsilon = 5\%$)
slightly darkens many pixels,
but the difference from the original is barely visible.
%
The figure further illustrates that,
among all techniques,
the explanations remain similar
if computed from a NN with one or four layers,
suggesting that even restricted encoding of the network
(e.g. \fixall{})
does not harm the generality.

\NEW{
Compared to \Verix{},
the explanations computed using our technique result in more complex formulas,
potentially occupying megabytes rather than kilobytes of memory.
However,
this does not hinder their interpretability:
for instance,
computing the minimum norm point in Fig.~\ref{fig:Norm_points_fixall} requires only a few seconds.
Further details are provided in
\ifSpringer
the extended version~\cite{extended_version}.
\else
Appendix~\ref{apx:size}.
\fi
}

\subsection{Modular Encoding Instantiation}
\label{sec:exp:chc}

The core intuition behind a modular translation lies in treating a neural network not as a
monolithic mathematical function, but as a discrete-time transition system. In this framework, the
activation values of neurons in a given layer represent the ``state'' of the system. The transition from one
layer to the next is governed by a recurring logical predicate that encapsulates the linear transformation
(weights and biases) and the non-linear activation function (\Relu{}).
By defining a generic transition formula, we decouple the architectural definition from the specific
instance of the network. This modularity allows formal verification tools to 
\emph{discover inductive invariants} that describe the properties of neurons in any hidden layer that can also serve as explanations\footnote{The main benefit of this non-monolithic encoding is providing better control over the shapes of the computed explanations.}.  
The purpose of this section is to demonstrate that this alternative encoding benefits from neuron activation,
similarly to when the monolithic encoding is used.

A modular transition $\psi_{step}$ captures the logic of a single layer transition from state $\mathbf{x}^{(\ell)}$ to $\mathbf{x}^{(\ell+1)}$ for any hidden layer $0 \leq \ell < L$:

\begin{equation}
\begin{aligned}
\label{4}
\psi_{step}(\mathbf{x}, \mathbf{x}', \ell, \ell') = &\ 0 \leq \ell < L \land \bigwedge_{i=1}^{n} \Big( x_i' = \mathit{ite}\big(  (\sum_{j=1}^{n} w^{(\ell)}_{j,i} x_j + b^{(\ell)}_i) > 0, \\
& \sum_{j=1}^{n} w^{(\ell)}_{j,i} x_j + b^{(\ell)}_i, 0 \big) \Big) \land \ell'  = \ell + 1
.
\end{aligned}
\end{equation}

An inductive invariant $Inv$ is a property of neurons in each hidden layer. Given the initial condition
$\varphi_s \land \DomainPhi$ and the error state $\neg\psi_c$, $Inv$ intuitively serves as a symbolic summary of activations that never reach the misclassification region, i.e., the three conditions hold:

\begin{enumerate}
    \item \textbf{Initiation:} $\varphi_s \land \DomainPhi \implies Inv(\mathbf{x}, 0)$
    \item \textbf{Inductiveness:} $Inv(\mathbf{x}, \ell) \land \psi_{step}(\mathbf{x}, \mathbf{x}', \ell, \ell') \implies Inv(\mathbf{x}', \ell')$
    \item \textbf{Safety:} $Inv(\mathbf{x},L) \implies \psi_c$
\end{enumerate}

We adopted the winning strategy from Sect.~\ref{sec:exp:smt} \fixall{} using the precomputed values
$A^{(\ell)}_i$ of neuron pre-activations for every layer $\ell$ and neuron $i$ based on evaluating the neural network on the sample.
The modified modular transition $\psi_{step}'$ uses these values to pick only the active or inactive phase:

\begin{equation}
\begin{aligned}
\label{5}
\psi_{step}'(\mathbf{x}, \mathbf{x}', \ell, \ell') = &\ 0 \leq \ell < L \land \bigwedge_{i=1}^{n} \Big(x_i' = \mathit{ite}\big(
A^{(\ell)}_i > 0, \\
& \sum_{j=1}^{n} w^{(\ell)}_{j,i} x_j + b^{(\ell)}_i, 0 \big) \Big) \land \ell'  = \ell + 1
.
\end{aligned}
\end{equation}

As a proof of concept, we implemented the encoding on top of an existing solver~\cite{FKB17,FB18},
which
iteratively guesses template instantiations (e.g., as defined below) and checks the initiation and inductiveness conditions.
The conjunction of instantiations that passed these checks 
as well as the safety condition forms a computed explanation.
%
The templates have the following forms: $(\ell=\_ \implies x_i^\ell \leq \_)$, $(\ell=\_ \implies x_i^\ell \geq \_)$, $(\ell=\_ \implies x_i^\ell = \_)$, $ (x_i^\ell \leq \_)$, $(x_i^\ell \geq \_)$, $(x_i^\ell = \_)$, where each $\_$ is a placeholder.
Because the behavior of the invariant synthesis depends heavily on a smart choice of the templates, our implementation does not aim to compete with the monolithic approach.
Instead,
we evaluate the impact of the fixing strategy by comparing the invariant synthesis runtimes
when using
$\psi_{step}$ (the \naive{} strategy)
and $\psi_{step}'$ (the \fixall{} strategy).
%
%
%
%
The evaluation
follows the setup in Sect.~\ref{sec:exp:smt},
using the medical benchmarks\footnote{
    The current prototype invariant synthesis
    is driven by template selection
    and
    is sensitive to the input feature count,
    timing out on all the image-recognition benchmarks.
}.
Fig.~\ref{fig:runtime_compare_CHC} illustrates a clear advantage of computing explanations using the \fixall{} strategy (red lines) over the \naive{} strategy (blue lines),
which mostly times out.
Overall, for both monolithic and non-monolithic encoding of NNs, the data confirms that the fixing is a viable solution to computing logical explanations
significantly faster.



\begin{figure}[t!]
    \centering
    \includegraphics[trim={0.35cm 0.4cm 1cm 0.1cm},clip, width=0.815\linewidth]{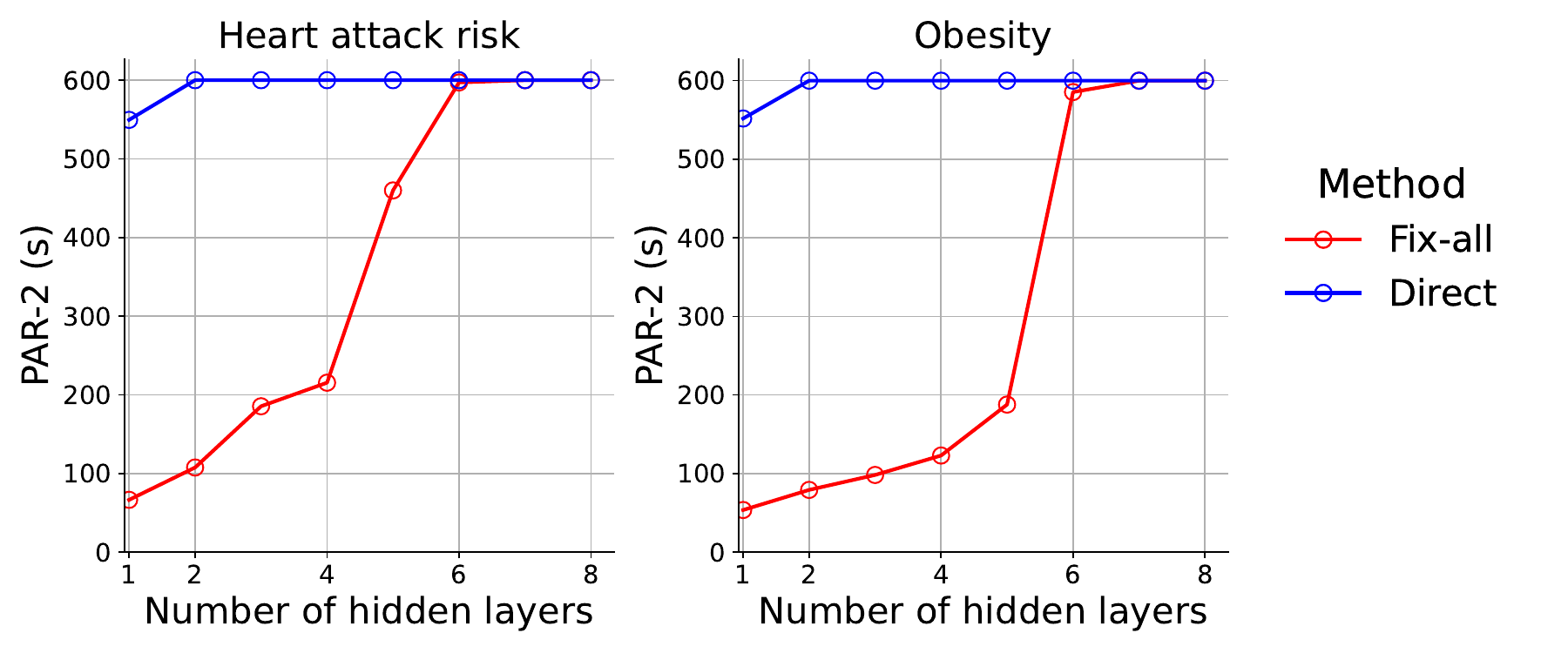}
    \caption{Performance of an enumerative invariant discovery with/without fixing.}
    \label{fig:runtime_compare_CHC}
\end{figure}

\section{Conclusion and Future Work}
\label{sec:conclusion}

This paper presented
a novel neuron activation-based symbolic framework for automated computation of logical explanations
that
\begin{enumerate*}[label=(\arabic*)]
    \item
    enables scaling to deep architectures
    not amenable to prior logic-based explainability methods,
    \item
    improves the generality of
    the resulting explanations
    that
    are provably correct
    and
    not restricted in shape,
    as opposed to explanations generated by prior methods that rely on specialized NN verifiers,
    \item
    isolates
    the NN-specific encoding
    from the reasoning algorithms
    that are based on general-purpose logical solvers,
    instantiated to \Opensmt{},
    utilizing the newly introduced
    transferable mechanism of
    guiding constraints,
    and
    \item
    is agnostic to the backend solver
    and
    the techniques
    that generalize the constraints on the sample point.
\end{enumerate*}
The experiments on diverse real-world benchmarks
from the domains of image recognition (GTSRB, CIFAR10, MNIST) and medicine
across a spectrum of complexities
showed that the new method
computes more general explanations
and
outperforms prior state-of-the-art approaches, often by orders of magnitude,
particularly the leading abduction-based tool for NNs \Verix{}.
The experiments also showed that the approach is generic
and applies beyond SMT encoding.

In future work,
we will aim to handle other NN structures, such as
convolutional NNs,
addressing more complex yet efficient classification tasks,
with no need for changes in our method at the theoretical level.
Moreover,
we will aim at embedding other activation functions beyond \Relu{},
such as Leaky \Relu{}.

\ifSpringer
\Acknowledgement{}
\fi

\bibliographystyle{splncs04}
\bibliography{references}

\ifSpringer
\else
\appendix

\section{Running Example}
\label{apx:example}


This section illustrates a direct SMT encoding (cf. Sect.~\ref{sec:encoding}) of the explainability problem,
using a toy example of a classifying neural network,
as well as modified encodings with the guidance driven by neuron activations (cf. Sect.~\ref{sec:activation-based}).

Fig.~\ref{fig:nn_example} shows a neural network (cf. Sect.~\ref{sec:background}) that maps 3-dimensional inputs $\mathbf{x}=(x_1, x_2, x_3)$ to two classes $\{c_1,c_2\}$
with input feature domains $\mathcal{D}_i = [0,4]$
(therefore, the feature space corresponds to a cube with edge length 4).
The network consists of 3 input neurons, 2 hidden neurons with $\Relu$ activation, and 2 output neurons.
The weights are set according to the edge labels in the figure and the biases are set to zero,
resulting in the following pre-activations:
\begin{alignat*}{6}
z_{1}^{(1)}(\FVec) &= 2x_1 + x_3
\qquad
&&
z_{1}^{(2)}(\FVec) = a_{1}^{(1)}(\FVec) - 4a_{2}^{(1)}(\FVec) 
\\
z_2^{(1)}(\FVec) &= -x_1 + x_2 -x_3
\qquad
&&
z_{2}^{(2)}(\FVec) = -a_{1}^{(1)}(\FVec) + 4a_{2}^{(1)}(\FVec) 
\end{alignat*}

If $a_{1}^{(2)}(\FVec) \geq a_{2}^{(2)}(\FVec)$, then $\kappa(\FVec) = c_1$
(i.e., the input is classified as class~$c_1$), otherwise, $\kappa(\FVec) = c_2$.
For example, for input sample $\Sample = (1,1,3)$, $a_{1}^{(2)}(\Sample) = 5$ and $a_{2}^{(2)}(\Sample) = -5$,
and therefore, sample~$\mathbf{s}$ is classified as class~$c_1$.

\begin{figure}[t!]
    \centering
    \includegraphics[width=0.4\linewidth]{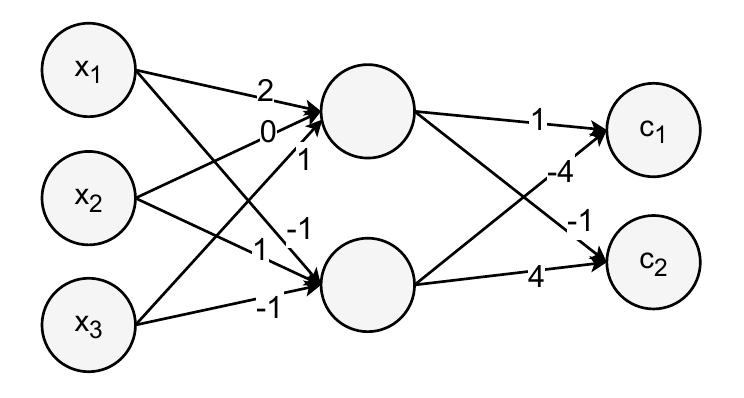}
    \caption{A simple neural network example}
    \label{fig:nn_example}
\end{figure}

\Paragraph{Direct Encoding}

Given the neural network above,
formula $\psi = \ClassifierPhi \land \DomainPhi \land \neg \ClPhi$
(cf. Sect.~\ref{sec:encoding}),
where $c = c_1$,
is encoded over the input feature variables $x_1, x_2, x_3$ as follows:
\begin{itemize}
\item the output neurons are encoded by
\[
\begin{aligned}
\ClassifierPhi := {} &
\bigl(
  a_{1}^{(2)} =
    \mathit{ite}(2x_1 + x_3 > 0,\; 2x_1 + x_3,\; 0)
\\[-0.3em]
&\qquad\qquad
  - 4 \cdot \mathit{ite}(-x_1 + x_2 - x_3 > 0,\; -x_1 + x_2 - x_3,\; 0)
\bigr)
\\
&\land
\bigl(
  a_{2}^{(2)} =
    -\, \mathit{ite}(2x_1 + x_3 > 0,\; 2x_1 + x_3,\; 0)
\\[-0.3em]
&\qquad\qquad
  + 4 \cdot \mathit{ite}(-x_1 + x_2 - x_3 > 0,\; -x_1 + x_2 - x_3,\; 0)
\bigr)
\end{aligned}
\]
using the auxiliary variables $a_{1}^{(2)}, a_{2}^{(2)}$
and inlining $a_{1}^{(1)}, a_{2}^{(1)}$
and $z_{1}^{(1)}, z_{2}^{(1)}$, $z_{1}^{(2)}, z_{2}^{(2)}$.

\item
the domains are encoded by
\[
\DomainPhi := (0 \leq x_1 \leq 4) \land (0 \leq x_2 \leq 4) \land (0 \leq x_3 \leq 4)\\
,\]

\item
the negation of classifying to class~$c = c_1$,
that is, classifying to class $c_2$,
is encoded by
\[
\neg \ClPhi := a_{1}^{(2)} < a_{2}^{(2)} .
\]
\end{itemize}

The input sample $\mathbf{s} = (1,1,3)$ classified as $c_1$
(i.e., $\kappa(\Sample) = c_1$)
is encoded as
\(
\SampleExpl = (x_1 = 1) \land (x_2 = 1) \land (x_3 = 3)
\).
Formula~$\psi$ is satisfiable,
while $\psi \land \SampleExpl$ is unsatisfiable.

Using the input configuration $\FixConf = \emptyset$
(i.e., with no fixing of activations of hidden neurons),
Alg.~\ref{Alg:Encode:Fix} (cf. Sect.~\ref{sec:activation-based})
outputs $\ConfigPhi \mapsto \top$ and a formula $\ClassifierPhi$
similar to the one above except it would also use auxiliary variables $a_{1}^{(1)}, a_{2}^{(1)}$
and not inline them.

\Paragraph{Neuron Activation-based Encoding}

This part shows examples of outputs of the algorithms in Sect.~\ref{sec:activation-based}.
Using the input configuration $\FixConf = \HiddenNSet$
(i.e., fixing all activations of hidden neurons)
and the sample $\Sample = (1,1,3)$,
Alg.~\ref{Alg:Encode:Fix} outputs
\[
\ClassifierPhi \mapsto \big( a_{1}^{(2)} = (2x_1 + x_3) - 4 \cdot 0 \big) \land \big( a_{2}^{(2)} = -(2x_1 + x_3) + 4 \cdot 0 \big)
\]
and
\[
\ConfigPhi \mapsto (2x_1 + x_3 > 0) \land (-x_1 + x_2 - x_3 \leq 0)
.
\]

Using the input configuration $\GuideConf = \HiddenNSet$
(i.e., guiding all activations of hidden neurons)
and the sample $\Sample = (1,1,3)$,
Alg.~\ref{Alg:Encode:Guide} outputs
guiding constraints
$\big( (2x_1 + x_3 > 0), (-x_1 + x_2 - x_3 \leq 0) \big)$.
Recall that Algs.~\ref{Alg:Encode:Guide} and~\ref{Alg:Encode:Fix}
are orthogonal,
although guiding a neuron that has already been fixed
is useless.

\section{Further Performance Experiments}
\label{apx:performance}

\begin{figure}
        \includegraphics[width=0.95\linewidth]{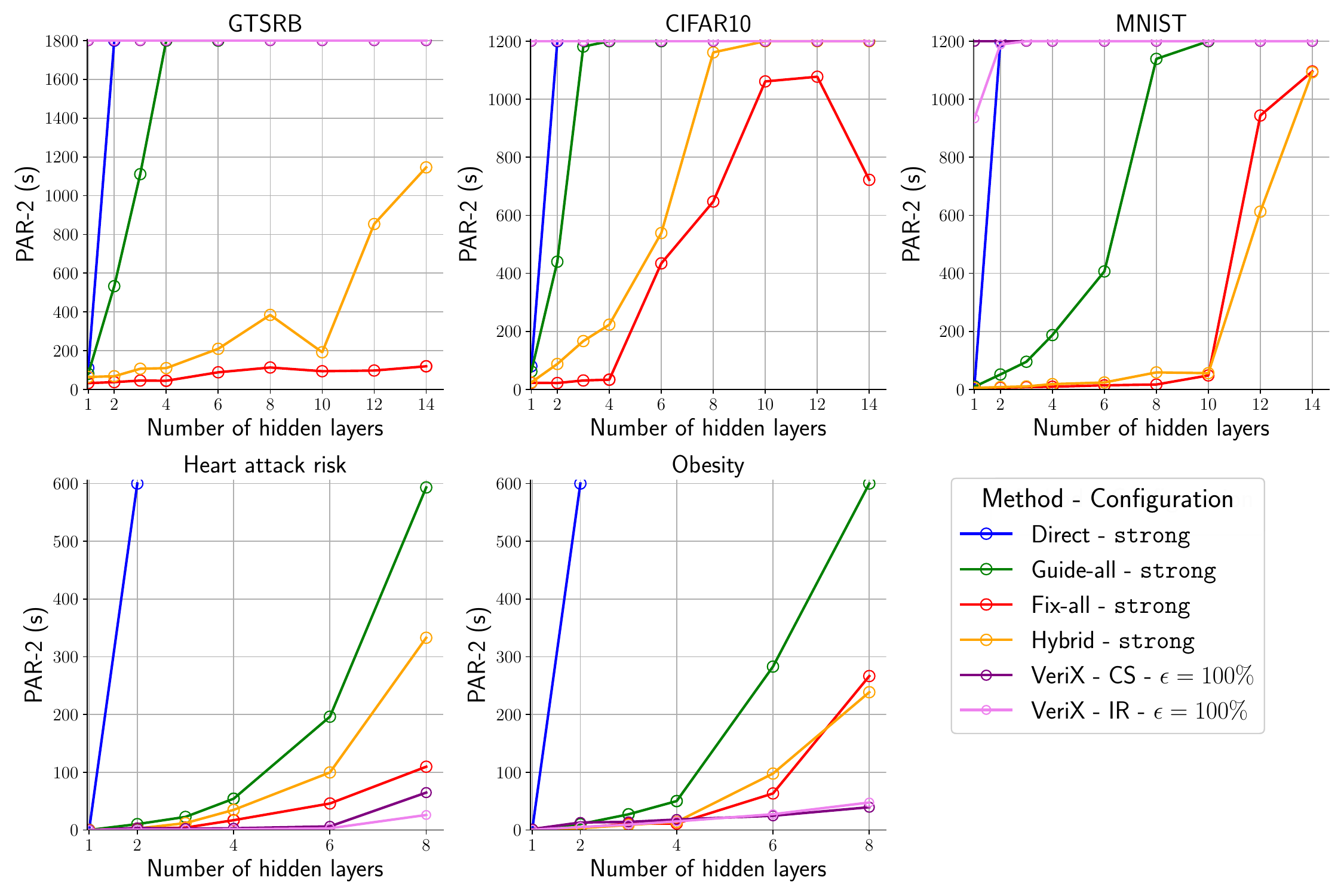}
    \caption{Performance comparison with \Verix{} on networks with 200 neurons per layer.}
    \label{fig:verix_compare200}
\end{figure}

Fig.~\ref{fig:verix_compare200} compares the performance of the same methods as in Fig.~\ref{fig:verix_compare}, selecting just the performant option $\epsilon=100\%$ for \Verix{}, on more complex networks with 200 neurons per hidden layer.
While some observations from Fig.~\ref{fig:verix_compare} remain similar, these experiments are much more challenging and only some of the methods manage to compute explanations for all sizes of the networks. Notably, \fixall{} still scales very well for GTSRB.
\Verix{} fails to compute the explanations in most image classification cases, while it performs well for medical datasets.


\section{Size of Explanations}
\label{apx:size}

\begin{figure}
    \centering
    \includegraphics[width=\linewidth]{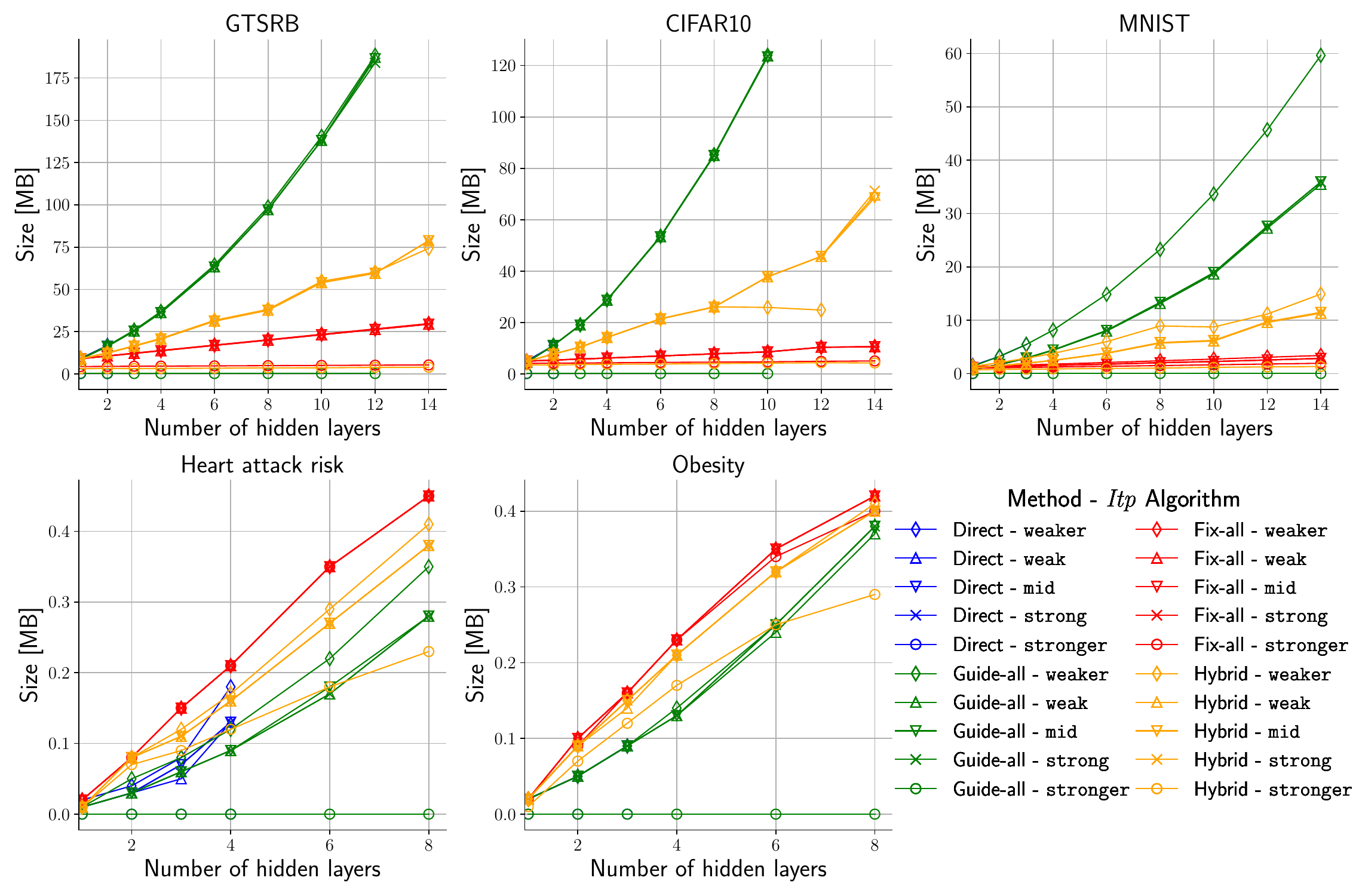}
    \caption{Explanation size comparison}
    \label{fig:plot:size}
\end{figure}

Our explanations are much more complex formulas compared to \Verix{}.
The explanations produced by \Verix{} take few KB of data, while the explanations by our method take up to tens or hundreds of MB per explanation depending on the size of the model and input.
Fig. \ref{fig:plot:size} shows the relation between the depth of the model and the explanations size.
For image classification tasks, \guideall{} produces bigger formulas, except for the \astronger{} interpolation algorithm, which is very restricted and in most cases represents a single sample point.

\fi

\end{document}